\documentclass[11pt,fleqn,reqno]{amsart}
\usepackage[numbers]{natbib}
\usepackage{hyperref}

\usepackage{amsfonts,amssymb,amsthm,amsmath}
\usepackage{epsfig}
\usepackage{graphicx}
\usepackage{mathrsfs}
\usepackage{color}

\usepackage{todonotes}
\usepackage{soul}

\usepackage[margin=1.1in]{geometry} 

\theoremstyle{plain}
\newtheorem{theorem}{Theorem}
\newtheorem{proposition}[theorem]{Proposition}
\newtheorem{lemma}[theorem]{Lemma}
\newtheorem{corollary}[theorem]{Corollary}

\newtheorem{rem}[theorem]{Remark}

\numberwithin{theorem}{section}
\numberwithin{equation}{section}

\newcommand{\nc}{\newcommand}

\nc{\be}{\begin{equation}}
\nc{\la}{\label}
\nc{\ba}{\begin{array}}
\nc{\ea}{\end{array}}
\nc{\bs}{\begin{split}}
\nc{\es}{\end{split}}

\newcommand{\R}{\mathbb{R}}
\newcommand{\C}{\mathbb{C}}
\newcommand{\Z}{\mathbb{Z}}

\newcommand{\frach}{\mathfrak{h}}

\newcommand{\cD}{\mathcal{D}}
\newcommand{\cC}{\mathcal{C}}
\newcommand{\cE}{\mathcal{E}}
\newcommand{\cF}{\mathcal{F}}
\newcommand{\F}{\mathcal{F}} 

\newcommand{\cH}{\mathcal{H}}

\newcommand{\cL}{\mathcal{L}}
 
\newcommand{\cP}{\mathcal{P}}         

\newcommand{\cS}{\mathcal{S}}

\nc{\e}{\epsilon}
\nc{\eps}{\epsilon}
\nc{\al}{\alpha}
\nc{\del}{\delta}
\nc{\h}{\delta}
\nc{\G}{\eta} 
\nc{\et}{\eta} 
\nc{\Gam}{\eta}  
\nc{\g}{\gamma}
\nc{\gam}{\gamma}
\nc{\ka}{\kappa}
\nc{\lam}{\lambda}
\nc{\Lam}{\Lambda}
\nc{\Om}{\Omega}
\nc{\om}{\omega}

\nc{\ta}{\tau}
\nc{\w}{\omega}
\nc{\io}{\iota}
\nc{\z}{\zeta}
\nc{\s}{\alpha}
\nc{\Si}{\Sigma}
\nc{\vphi}{\varphi}

\newcommand{\ve}{\varepsilon}
\newcommand{\veps}{\varepsilon}

\nc{\bP}{\bar{P}}
\nc{\bQ}{\bar{Q}}
\nc{\x}{\underline{x}}
\nc{\y}{\underline{y}}

\nc{\ran}{\rangle}
\nc{\lan}{\langle}

\newcommand{\ra}{\rightarrow}

\newcommand{\ls}{\lesssim}
\newcommand{\gs}{\gtrsim}
\newcommand{\one}{\mathbf{1}}

\renewcommand{\Re}{\mathrm{Re}} 
\newcommand{\Tr}{\mathrm{Tr}}
\newcommand{\tr}{\mathrm{Tr}}
\newcommand{\diag}{\mathrm{diag}}

\nc{\bfone}{{\bf 1}}
\nc{\1}{{\bf 1}}

\newcommand{\p}{\partial}

\newcommand{\n}{\nabla}

\newcommand{\curl}{\operatorname{curl}}
\newcommand{\CURL}{\operatorname{curl}}
\newcommand{\divv}{\operatorname{div}}

\renewcommand{\div}{\operatorname{div}}
\newcommand{\grad}{\operatorname{grad}}
\newcommand{\hess}{\operatorname{Hess}}

\newcommand{\na}{\nabla_a}

\newcommand{\COVGRAD}[1]{\nabla_{\!\!#1}}

\newcommand{\rbrac}[1]{\left(#1\right)} 

\def\eqn{\begin{align}}
\def\eeqn{\end{align}}

\newcommand{\DETAILS}[1]{}

\def\qf{\varphi}
\def\ab{a_b}
\def\chis{\chi_s^\lat}
\def\chib{\chi_s^b}
\def\chibt{\chi_t^b}
\newcommand{\frachb}{\mathfrak{h}^{\lat}}
\newcommand{\frachbvec}{\vec{\mathfrak{h}}_{b}}
\newcommand{\frachvec}{\vec{\mathfrak{h}}}

\newcommand{\hga}{h_{\g a}}

\newcommand{\lat}{\mathcal{L}} 
\newcommand{\Omd}{\Om}
\newcommand{\den}{\text{den}}
\newcommand{\sech}{\text{sech}}

\begin{document}

\begin{quote}
\qquad \qquad \qquad  \qquad \qquad \qquad   \qquad \qquad \qquad 	{\it To appear in Advances in Mathamtics \\ }
\end{quote}


\title[Bogoliubov-de Gennes Equations June 19, 2019]{Vortex lattices and the Bogoliubov-de Gennes Equations} 

\date{April 8, 2020} 
{\center\author{Ilias (Li) Chenn and I. M. Sigal}} 
\maketitle

\begin{abstract} We consider the Bogoliubov-de Gennes equations giving an equivalent formulation of the BCS theory of  superconductivity.  We are interested in static solutions with the magnetic field present. We carefully formulate the equations in the basis independent form, discuss their general features and isolate  key physical classes of solutions (normal 
 and vortex lattice states) which are the candidates for the ground state. We prove existence of the normal 
  and vortex lattice states and stability of the normal states for large  temperature or magnetic fields and their instability for small  temperature and small magnetic fields. 
  
Mathematics Subject Classification (MSC): 81Q80, 81V70, 81V74, 35Q40, 35Q56, 35Q82 

\end{abstract}


\section{Introduction} 

\subsection{Background}\label{sec:backgr}

The Bogoliubov-de Gennes (BdG) equations describe the remarkable quantum phenomenon of  superconductivity.\footnote{For some physics background, see books \cite{dG, MartRoth} and the review papers \cite{Cyr, Legg}.} They present an equivalent formulation of the BCS theory of superconductivity and are among the latest additions to the family of important effective equations of mathematical physics. Together with the Hartree, Hartree-Fock (-Bogoliubov), Ginzburg-Landau, Gross-Pitaevskii and Landau-Lifshitz  equations, they are the quantum members of this illustrious family consisting of such luminaries as 
the heat, 
 Euler, Navier-Stokes, Boltzmann and Vlasov equations. 

The BdG equations describe the evolution of superconductors on nanoscopic and macroscopic scales. There are still many fundamental questions about these equations which are completely open, namely 
\begin{itemize}
\item 
Derivation;

\item Well-posedness;

\item Existence and stability of stationary magnetic solutions. 
\end{itemize}
 
In this paper we address 
the third problem.\footnote{A formal derivation of the (time-dependent) BdG equations is discussed in Appendix \ref{sec:qf-reduct}. For a different but related formal derivation (going back to Dirac and Frenkel)  see \cite{BenSokSol, Lub}. Presently, there are no rigorous derivations. For a recent book on 
 rigorous derivations of the simpler Hartree-Fock equation see \cite{BenPorSchl}, for some more recent papers \cite{PortRadSaffSchl, PetPick} and for the related bosonic Hartree-Fock-Bogoliubov equations \cite{GrillMached1, GrillMached2, NamNap}. For the literature before 2016, we refer only to the papers \cite{He, Sp}, which led to the most of the recent developments and to the paper \cite{FKP}, dealing with the relation between the mean-field limit and deformation quantization.   Following appearance of the first version of this paper on arXiv, the existence problem in the absence of magnetic field was taken up in \cite{BenSokSol} (cf. \cite{BBCFS} for the global existence for the related bosonic Hartree-Fock-Bogoliubov equations).}  
By the magnetic solutions we  mean 
solutions with non-zero magnetic fields. 

The key special solutions of the Bogoliubov-de Gennes (BdG) equations are  normal, superconducting and mixed or intermediate states. The superconducting (or Meissner) states assume, by the definition, that the magnetic field is zero, while  mixed ones, to have 
non-vanishing magnetic fields. For type II superconductors, according to experiments, the latter are 
 (magnetic) vortex lattices. In this paper, we  prove the existence of the normal states for non-vanishing magnetic fields and partial results on their stability and the existence of the vortex lattices. 

There is a considerable physics literature devoted to the BdG  equations, but despite the role played by magnetic phenomena in superconductivity it deals mainly with the zero magnetic field case, with only few disjoint remarks about the case when the magnetic fields are present, the main subject of this work.\footnote{The Ginzburg-Landau equations give a good account of magnetic phenomena in superconductors but only for temperatures sufficiently close to the critical one (see e.g. \cite{Cyr, dG, Tink}).} 

As for rigorous work, it also deals exclusively with the case of zero magnetic field.  The general (variational) set-up for the static BdG equations is given in  \cite{BLS}. 
The next seminal works on the subject are  \cite{HHSS}, where the authors prove the existence of superconducting states (the existence of the normal states under the assumptions of \cite{HHSS} is trivial), to which our  work is closest, and \cite{FHSS},  deriving the (macroscopic) Ginzburg-Landau equations. 
 For an excellent, recent review of 
 the subject, 
 with extensive references and discussion see \cite{HaiSei}. 
 
In the rest of this section we introduce the BdG  equations, describe their properties and the main issues and present the main results of this paper. In the remaining sections 
 we prove these results, 
 with technical derivations delegated to appendices. 
  In the last appendix, following \cite{BBCFS}, we discuss a formal, but  natural, derivation of the BdG  equations.

\subsection{Bogoliubov-de Gennes (BdG) equations}\label{sec:bdg}

In the Bogoliubov-de Gennes approach states of superconductors are described by the pair of bounded operators $\g$ and $\s$, acting on the  one-particle complex Hilbert space, $\frach$,  with a complex conjugation (an anti-linear involution), and satisfying 
\begin{equation}  \label{gam-al-prop} 
	0\le \gamma=\g^* \le 1,\  \quad 	 \s^*= \overline\s\ \quad  
	  {\rm\ and\ }\  \quad \alpha \alpha^* \leq \gamma(1-\gamma)
\end{equation}
where  $\overline\gamma :=\cC \g \cC$, with $\cC$, the operation of complex conjugation 
 (see Appendix \ref{sec:qf-reduct} for the origin of these operators).  $\g$ is a one-particle density operator 
 and $\s$ is a two-particle coherence operator, or diagonal and off-diagonal correlations. $\g(x, x)$ is interpreted as the one-particle density. 

  The one-particle space $\frach$ is determined by the many-body quantum problem. 
For zero density (or `finite') systems, it is  $L^2(\R^d)$ and for positive density ones,  $L^2(\Om)$, where $\Om$ is an arbitrary fundamental cell in a lattice $\lat\subset \R^d$, with magnetic field dependent (twisted) boundary conditions (see Subsection \ref{sec:1pspaces} for more details).  To fix ideas, we take $\frach$ to be 
  the latter, specifically, 
 \begin{align}\notag 
	\frach  := & \{ f \in L^2_{loc}(\R^2) : u_{s}^\lat f = f \text{ for all } s \in \lat \},
\end{align}
 where $u_{ s}^\lat$ are the magnetic translations given in \eqref{ubs} below,  with the scalar product given by $L^2(\Om)$ for an arbitrary fundamental cell $\Omd$ of $\lat$. Furthermore, we understand $\int$ without specifying the domain of integration as taken over $\Om$. 
\DETAILS{$L^2(\R^d)$, if $X = \R^d$, and the trace per volume $\Om$,
\begin{align} \label{Tr-per-vol}
	\Tr_\Om(A) := \frac{1}{|\Om|}\Tr_{L^2(\R^3)}(\chi_\Om A \chi_\Om),
\end{align}	
where $\chi_\Om$ is the indicator function on $\Om$,
 if $X = \Om$ is a fundamental cell of a lattice $\lat$.}

The BdG equations form a system of coupled, nonlinear equations for  $\g$ and $\s$. 
It is convenient to organize the operators  $\gamma$ and $\s$ 
 into the self-adjoint operator-matrix 
\begin{align} \label{Gam}
	\G:= \left( \begin{array}{cc} \g & \s \\  \s^* & \one-\bar\g \end{array} \right).
\end{align}
The definition of $\g$ and $\al$ in terms of the many-body theory (see \eqref{omq-mom} of Appendix \ref{sec:qf-reduct}) implies that 
\begin{align} \label{Gam-prop} 
	0\le \G=\G^* \le 1\ {\rm\ and\ }\   J^* \G J=\one-\bar\G,\ J:=\left( \begin{array}{cc} 0 & \one \\  - \one & 0 \end{array} \right). \end{align} 
These relations imply relations \eqref{gam-al-prop}. 

Since  the BdG equations describe the phenomenon of superconductivity, they are naturally coupled to the electromagnetic field. We describe the latter by the vector  potential $a$ in the {\it gauge in which the electrostatic potential is zero} so that the equations have a slightly simpler form). 
In what follows, for an operator $A$, we denote by $A(x, y)$ its integral kernel. Then the time-dependent BdG equations state (see e.g. \cite{dG, Cyr, MartRoth}) 
\begin{align}\label{BdG-eq-t}
& i \partial_t  \G =[\Lambda(\G, a), \G],\\ 
\label{Lam} 
&\Lambda(\G, a) =  \big(\begin{smallmatrix} h_{\g a} & v^\sharp \s\\ v^\sharp \bar{ \s} & -\overline{h_{\g a}}\end{smallmatrix}\big),\ \quad  
 h_{\g a}  =-\Delta_a+ 
  v^* \gamma  - v^\sharp\, \gamma \,, 
\end{align} 
where, for the pair potential $v(x, y)$, the operator $v^\sharp$ 
(acting on operators on $\frach$)  is defined through the integral kernels as $(v^\sharp\, \al) \, (x;y):= v(x, y)\al (x;y)$  and 
$(v^* \gamma)(x) :=\int v(x, y) \rho_\gamma(y) d y,$
with $\rho_\gamma(x):= \gamma(x; x)$.  
 $v^* \gamma$ and $v^\sharp\, \gamma$ are the direct and exchange self-interaction potentials, and  $\Delta_a:=|\n_a|^2, \n_a:=\n- i a$. 
  \eqref{BdG-eq-t} is coupled to the  Maxwell equation (Amp\`{e}re's law)
\begin{align}\label{Amp-Maxw-eq}
 \partial_t^2 a   &= -\curl^* \curl a + j(\G, a),
\end{align}
where $j(\G, a)(x)\equiv j(\g, a)(x) := [-i \na, \g]_+(x, x)$ 
 is the superconducting current.  Here $[A, B]_+$ is the anti-commutator of $A$ and $B$, i.e. $[A, B]_+:=AB+BA$. (To see that \eqref{Amp-Maxw-eq} is Amp\`{e}re's law, one recalls that in our gauge the electric field is $E=-\p_t a$.) In what follows, we assume that 
\begin{align}\label{v-cond}
	v(x; y)=v(x - y)\ \text{ and }\ v(-x)=v(x),
\end{align}
so that $v^* \gamma = v* \rho_\gamma $. We specify below additional assumptions on  $v(x,y)$ and on the operators $\g$ and $\al$  so that all the terms in  \eqref{BdG-eq-t}-\eqref{Amp-Maxw-eq} are well defined (at least weakly).

\medskip

 {\bf Connection with  the BCS theory. }  \eqref{BdG-eq-t} can be reformulated as an equation on the Fock space involving an effective quadratic Hamiltonian  (see \cite{Cyr, dG, HaiSei} and \cite{BBCFS}, for the bosonic version). These are the effective BCS equations and the effective BCS Hamiltonian. 
 
 {\bf Stationary equations.}  As was mentioned above, we are interested in stationary solutions to \eqref{BdG-eq-t}-\eqref{Amp-Maxw-eq}. These solutions satisfy the stationary BdG equations which, as we show below (see Proposition \ref{prop:stat-sol} and after), are of the form
\begin{align} \label{Gam-eq}
	&\Lam_{\G a} -T g'(\G) =0,\\ 
 \label{a-eq} 
     &\CURL^*\CURL a - j(\G, a)=0,\\
  \label{Lam'} 
& \Lam_{\G a }:=\left( \begin{array}{cc} h_{\g a \mu} & v^\sharp \s \\ v^\sharp \s^* & - \bar h_{\g a \mu} \end{array} \right), \end{align}
with $h_{\g a \mu}:=h_{\g a}-\mu$, where $\hga$ is given in \eqref{Lam}. 
Here $\mu$ and $T\ge 0$ are the chemical potential and temperature parameters. 
The physical function $g'$ is selected by either a thermodynamic limit (Gibbs states) or by a contact with a reservoir (or imposing  the maximum entropy principle) and  is given by
\begin{equation} \label{g'} 
	g'(\lam) = - \ln \frac{\lam}{ 1-\lam}. 
\end{equation} 
Its inverse, ${g'}^{-1}$, is the Fermi-Dirac distribution	
\begin{align} \label{FD-distr} 
 f_{\rm FD} (h) =(1+e^{ h})^{-1}.\end{align} 
By inverting $g'$ in \eqref{Gam-eq}, we can rewrite this equation as the nonlinear Gibbs equation: 
\begin{align}  \label{Gam-eq'}	& \G = f_{T}(\Lam_{\G a}),\  \text{ where }\  f_{T}(\lam):= f_{\rm FD}(\frac{1}{T} \lam).\end{align} 
In Remark \ref{rem:part-hole-sym} below, we rewrite \eqref{Gam-eq} in terms of the eigenfunctions of $\Lam_{\G a}$ or $\G$, the form common in physics literature and closer in spirit to the PDEs.

\subsection{Symmetries and conservation laws} 
\label{sec:bdg-prop} 
  \eqref{BdG-eq-t}-\eqref{Amp-Maxw-eq} are invariant under the time-independent {\it gauge} transformations, 
\begin{equation}\label{gauge-transf}
    T^{\rm gauge}_\chi : (\g, \s, a) \mapsto (e^{i\chi }\g e^{-i\chi } , e^{i\chi }  \s e^{i\chi }, a + \nabla\chi), 
\end{equation}
  for any sufficiently regular function $\chi :  \R^d \to \R$,  and  
the {\it translation}, {\it rotation} and {\it reflection} transformations, 
\begin{align}\label{tr-transf}
   & T^{\rm trans}_h :  (\g, \s, a) \mapsto (u_h\g u_h^{-1}, u_h\s u_h^{-1}, u_h a),\\
\label{rot-transf}
  &  T^{\rm rot}_\rho : (\g, \s, a) \mapsto (u_\rho\g u_\rho^{-1} , u_\rho\s u_\rho^{-1}, \rho u_\rho a),\\
\label{refl-transf}
  &  T^{\rm refl} : (\g, \s, a) \mapsto (u \g u^{-1} , u \s u^{-1}, - u a),\end{align}
 for any $h \in \R^d$ and $\rho \in O(d)$. 
Here $u_h\equiv u_h^{\rm trans}$, $u_\rho\equiv u_\rho^{\rm rot}$ and $u\equiv u^{\rm refl}$ are the standard translation, rotation and reflection transforms $u^{\rm trans}_h :  f(x) \mapsto f(x + h)$, $u^{\rm rot}_\rho : f(x)  \mapsto  f(\rho^{-1}x)$ and $u^{\rm refl} : f(x)  \mapsto  f(- x)$.

We also keep the notation $T^{\rm trans}_h$ for these operators restricted to $\G$'s.

  \eqref{BdG-eq-t}-\eqref{Amp-Maxw-eq}  conserve the energy  
    \begin{align} \notag \cE (\G, a) :=& \Tr \big((-\Delta_a) \gamma \big) + \frac12 \Tr \big((v* \rho_\g) \gamma \big)
   -\frac12\Tr \big((v^\sharp \g) \gamma \big) \\
 \label{energy-t} & 
+ \frac12 \Tr \big( \s^* (v^\sharp \s) \big)+ \frac12\int ( |\curl a|^2+ |\p_t a |^2), 
\end{align} 
where the trace is taken in the one-particle space $\frach$.  (Recall that in our gauge, $ -\p_t a$ is the electric field.)  
 Indeed, assuming $(\G, a)$ solves  \eqref{BdG-eq-t}-\eqref{Amp-Maxw-eq}, taking the time derivative of \eqref{energy-t}, using the expression \[d_a \Tr\big((-\Delta_a) \gamma \big)a'=-\Tr\big(j(\G, a) a' \big)\] for the G\^ateaux derivative in $a$ and using the notation $\dot f\equiv \p_t f$ and $\ddot f\equiv \p_t^2 f$ for 
 $f=\g, \al, a$, we find 
 \begin{align} \label{dtE}	\partial_t \cE (\G, a) 
=& \tr(h_{a \g} \dot \g)  +\Tr\big( \dot\s^* (v^\sharp \s) +  \s^* (v^\sharp \dot\s)\big)\\
 \label{dtE'}& + \lan \curl^* \curl - j(\G, a), \dot a\ran + \lan \ddot{a} , \dot a \ran,	
\end{align}
where the inner product 
is understood to be  in $L^2(\Omd)$.  A simple computation shows that line \eqref{dtE} 
$= \frac12 \tr(\dot\G \Lambda(\G, a))$ and therefore by \eqref{BdG-eq-t} is $0$. Line \eqref{dtE'} vanishes by \eqref{Amp-Maxw-eq}. Hence  
 $\partial_t \cE =0$. \qquad \qquad  \qquad \qquad  \qquad \qquad  \qquad \qquad $\Box$

Furthermore, the global gauge invariance implies the evolution conserves the number of particles, $N:=\Tr \g$.

 \subsection{Free energy} 
The stationary BdG equations arise as 
the Euler-Lagrange equations  for the free energy (BCS) functional 
\begin{align}\label{FT-def} F_{T}(\G, a):=E(\G, a) -T S(\G)-\mu N(\G),\end{align}
where $S(\G) = \Tr g(\G)$, with $	g(\lam):= -\lam \ln \lam - (1-\lam)\ln (1-\lam)$, an anti-derivative of \eqref{g'}, is the entropy, 
 $N(\G):=\tr \g$ is the number of particles, and $E(\G, a)$ is the energy functional \eqref{energy-t} for $\G$ and $a$  time-independent and is given by 
\begin{align} \notag 
	E(\G, a) &=\Tr\big((-\Delta_a) \gamma \big) +\frac12\Tr\big((v* \rho_\g) \gamma \big) -\frac12\Tr\big((v^\sharp \g) \gamma \big)\\
 \label{energy} & +\frac{1}{2} \Tr\big( \s^* (v^\sharp \s) \big) 
 +\int  |\curl a |^2. 
\end{align}

Note surprisingly, it turns out that $E(\G, a):=\qf(H_a)$, where  $\qf$ is a quasi-free state in question (see Appendix \ref{sec:qf-reduct}) and  $H_a$  is the standard many-body given in \eqref{H}, coupled to the vector potential $a$.

\bigskip

 From now on, we assume that $d=2$. This is a typical case considered in physics applications, where it is assumed 
  that the solutions are independent of the third coordinate (the cylindrical geometry in the physics terminology).

 Let $\ab(x)$ be the vector potential with the constant magnetic field, $\curl a_b=b$. 
 Below, we work in the symmetric gauge: $\ab(x)=\frac b2 (-x_2, x_1)$. 

 \subsection{Magnetic translations}\label{sec:mag-trans}

Let  $\cL\subset \R^2$ be a Bravais lattice, fixed throughout the paper. 
We define  the magnetic translation operator  
\begin{align}\label{ubs} u_{ s}^\lat := u^{\rm gauge}_{-\chis} u^{\rm trans}_{ s},\end{align} where 
$u^{\rm gauge}_\chi :    \phi(x) \mapsto e^{i\chi(x)}\phi(x)$  and, recall, $u^{\rm trans}_h :    \phi(x) \mapsto \phi(x + h)$ (cf. the definitions after \eqref{refl-transf}) and $\chi_\cdot^\lat: \lat\times \R^2\ra \R$ is given by 
\begin{align}\label{chibs}\chis(x):=\frac b2 (s\wedge x) 
+c_s,\ \quad c_s:=\frac b2 (s''\wedge s'), 
\end{align} 
 where $s=k\om_1 +\ell \om_2=:s'+s'', k, \ell\in Z$, for any $s \in \lat$, with $\{\om_1, \om_2\}$ a basis in $\lat$, and 
 $b$ is given by 
\begin{align}\label{b-quant}
   b=\frac{2\pi n}{|\lat|},\ n \in  \Z. 
\end{align} 
Here $|\lat|$ is the area of a 
 fundamental cell 
  of $\lat$ (which is independent of the choice of the cell). One can easily check that the operator-family $u_{ s}^\lat, s\in \lat,$  is a group representation: 
  \begin{align}\label{ubs-gr} u_{ s}^\lat u_{ t}^\lat= u_{ s+t}^\lat.\end{align}

 \subsection{One-particle spaces}\label{sec:1pspaces} With translations $u_{ s}^\lat$ given in \eqref{ubs}, we define the periodic one-particle state space 
\begin{align}\label{fh-lat}
	\frach  := & \{ f \in L^2_{loc}(\R^2) : u_{s}^\lat f = f \text{ for all } s \in \lat \},
\end{align}
 which is a Hilbert space with the scalar product defined, for an arbitrary fundamental cell $\Omd$ of $\lat$, as 
 \begin{align}\label{frachb-norm}\lan f, g\ran_{\frach} :=\lan f, g\ran_{L^2(\Omd)}. 
\end{align} 

\medskip

Similarly, we consider the Sobolev space of vector potentials:  let $\vec H^{r}$ be given by 
\begin{align}\label{fh-vec}
\vec H^{r} := & \{ a \in H^r_{\rm loc}(\R^2; \R^2) : T_s^{\rm trans} a = a\ \forall s \in \lat,\ \div a=0, \int a = 0\},
\end{align}
with the Sobolev norm $\|a\|_{H^r}\equiv \|a\|_{H^r(\Omd)}$, for some (and therefore for every) fundamental cell $\Om$ of $\lat$. 
  (The conditions $\div a=0$ and $\int a \equiv \int_\Om a  = 0$ make the operator $\curl^*\curl$ strictly positive.)	Finally, we define the affine space
\begin{align}\label{fhb-vec}\frachvec^{r} := a_b + \vec H^{r}. 
\end{align}
(Recall that $a_b(x)$ is the vector potential with the constant magnetic field.)
 
Now, we define spaces of $\g$'s and $\al$'s used below. Let $I^r$ denote the space of bounded operators satisfying $\|A\|_{I^r}:=(\tr(A^*A)^{r/2})^{1/r}<\infty$ (Schatten, or a trace ideal, or non-commutative $L^r-$space) and let $M_b := \sqrt{-\Delta_{a_b}}$.  
We define 
Sobolev-type spaces for trace class operators by 
\begin{align} \label{Is1}	 
 	& I^{s,1} := \{ A : \frach  \rightarrow \frach  :  
	  \|A\|_{I^{s,1}} := \|M_b^{s} A M_b^{s}\|_{I^{1}} < \infty \},\\
\label{Is2} 
	&  I^{s,2} := \{ A : \frach  \rightarrow \frach : \, 
	 \|A\|_{ I^{s,2}} := \| A M_b^{s}\|_{I^{2}} 	< \infty \}.
\end{align}
Note that $I^{0,p} = I^p$. We will usually assume $\g\in I^{1,1}$ and $\al\in  I^{1,2}$. (Strictly speaking $\al$ acts from the dual space $\frach^*$ to $\frach$ (see \cite{BLS} for details), but for the sake of notational simplicity we will ignore this subtlety, by identifying $\frach^*$ with $\frach$.) 
 We will use the  notation $\hat I^{s}$ for the space of $\G$'s on $\frach\oplus \frach$ with $\g\in I^{s,1}$ and $\al\in I^{s,2}$ and the norm
\begin{align} \label{eta-norm}	\|\G\|_{(s)} := \|\gamma \|_{I^{s,1}} + \|\al \|_{ I^{s,2}}.
\end{align} 
\DETAILS{	We denote by $I^{s,1, 2}$ 
	the Sobolev spaces of operators $\G$, with $\g\in I^{s,1}$ and $\al\in  I^{s,2}$, equipped with the norms
\begin{align}
	\|\G\|_{(s)} := \|\gamma \|_{I^{s,1}} + \|\alpha\|_{ I^{s,2}}
\end{align}}
Due to Lemma \ref{lem:al-gam-bnd}(2) below, $\g\in I^{1,1}\Rightarrow \al\in I^{1,2}$ and $\|\G\|_{(s)} \simeq \|\gamma \|_{I^{s,1}}$. 
 Furthermore, we define  
\DETAILS{\begin{align}
	\mathcal{D}^s := \mathcal{D} \cap I^{s,1, 2}, 
\end{align}
with $s=1$, where   {\bf(what is $\cL( \mathfrak{h} \times \mathfrak{h}^*)$?)} } 
\begin{align} \label{eqn:defDomainD}
	 \mathcal{D}^s_\nu 
	=\big\{\G\in \hat I^{s},\ 
	 \G \text{ satisfies \eqref{Gam-prop}},\  
	&\tr\g=\nu 
	\big\}. 
\end{align}

\subsection{
Ground states of BdG equations}
\label{sec:bdg-spec-sol}

The static BdG equations \eqref{Gam-eq}-\eqref{a-eq} have the following key classes of solutions which are candidates for the ground states:
\begin{enumerate}
\item Normal state:  $(\G, a)$, with $\al=0$ (i.e. $\G$ is diagonal). 

\item Superconducting state:  $(\G, a)$, with $\s \ne 0$ and $a = 0$.

\item Mixed state:  $(\G, a)$, with $\s \ne 0$ and $a \ne 0$.
\end{enumerate}

We discuss the above states in more detail. 

\medskip
\paragraph{\bf Superconducting states.} The existence of superconducting,  translationally invariant solutions is proven in \cite{HHSS}. (See  the second part of Subsection \ref{sec:backgr} and \cite{HaiSei} for the references to earlier results.) 

\medskip

\paragraph{\bf Normal states.}Since we assumed that the external fields are zero, the equations (considered in $\R^2$) are translationally invariant. Because of the gauge invariance, it is natural to consider the simplest, gauge (magnetically) translationally invariant solutions, i.e. solutions invariant under the transformations   
\begin{align}\label{mag-transl}T_{b s}: = (T^{\rm gauge}_{\chib})^{-1}T^{\rm trans}_{ s}, 
\end{align}  
for any $s\in \R^2$ and the function  $\chib (x), s, x\in \R^2$, given by  
\begin{align}\label{chibs'}\chib(x):= \frac b2 (s\wedge x). 
 \end{align} 

 For $b=0$, we can choose $a=0$. 
 In this case, the existence of  normal translationally invariant solutions was proven in \cite{Ha}.

For $b\ne 0$, the simplest normal states are  the magnetically translation (mt-) invariant ones,  i.e. ones satisfying
\begin{align}\label{mt-invar}T_{b s} (  \G, a) = (  \G, a),\end{align}  
for any $s\in \R^2$, where $T_{b s}$ is defined in \eqref{mag-transl}. 
 Here the operator $\G$ acts on $L^2_{\rm loc}(\R^2)\times L^2_{\rm loc}(\R^2)$.  

 
\medskip

\paragraph{\bf Mixed states.}
 
The main candidate for a mixed state is a {\it vortex lattice}, i.e. a state, $(\G, a)$, satisfying $\al\neq 0$ and the equivariance condition 
\begin{align}\label{VL:equiv}T^{\rm trans}_{ s} (\G, a) =  T^{\rm gauge}_{\chis} (\G, a)\  \text{ for every }\ s \in \lat, 
\end{align}
where $\chis: \lat\times\R^2 \to \R$ are defined in \eqref{chibs}, 
  with $b$ satisfying the quantization condition \eqref{b-quant}.  

\medskip

\begin{proposition}[Magnetic flux quantization]\label{prop:mf-quant}  
If a vector potential $a$ satisfies \eqref{VL:equiv}, 
 then for any fundamental cell $\Omd$ of $\lat$, we have
 \begin{align}\label{mf-quant} 
\frac{1}{2\pi} \int_{\Omd} \curl a = c_1(\chi^\lat)\in \Z, 
\end{align} 
where 
 $c_1(\chi^\lat)$, called the first Chern number  of $\chis(x)$, 
 is the integer $n$ entering  \eqref{chibs}-\eqref{b-quant}. 
 \DETAILS{defined, for any basis $\{\nu_1, \nu_2\}$ in $\cL$, as (see Appendix \ref{sec:MT} for more details)   {\bf($\chib(x)$?)} \begin{equation}\label{c1}
    c_1(\chi) =\frac{1}{2\pi} (\chi_{\nu_2}(x+\nu_1) - \chi_{\nu_2}(x) - \chi_{\nu_1}(x+\nu_2) + \chi_{\nu_1}(x)), 
\end{equation}}   
 \end{proposition}
 One can show that $\chi^\lat=\chis(x)$ is a co-cycle and $c_1(\chi^\lat)$ can be defined in terms of the function $\chi^\lat$ using its co-cycle property and without reference to the explicit form \eqref{chibs} of $\chi^\lat$  used, 
 see Remark \ref{rem:co-cycle} for the definitions and a discussion.  
 

 Recall that $\ab(x)$ is the magnetic potential with the constant magnetic field $b$ ($\curl a_b=b$). In what follows we fix the gauge in which $\ab(x)=\frac b2 *x, *x:=(-x_2, x_1)$. Then  $\ab$ satisfies \eqref{VL:equiv}, provided $b$ satisfies the quantization condition \eqref{b-quant}. 

\subsection{Results} \label{sec:results}   
 Let the reflection operator $\tau^{\rm refl}$ be given by conjugation by the reflections, $u^{\rm refl}: f(x)\ra f(-x)$. We say that a state $(\G, a)$ is {\it even (or reflection symmetric)} if and only if 
 \begin{align} \label{even} \tau^{\rm refl} \g =\g,\  \tau^{\rm refl} \al =\al\ \text{ and }\ u^{\rm refl} a = - a.
\end{align} 
The reflection symmetry of the BdG equations implies that if an initial condition is even then so is the solution at every moment of time.

In what follows, we use the notation $A\ls B$ and $B\gs A$ to signify the inequalities $A\le C B$ and $B\ge c A$, where $C$ and $c$ are positive constants independent of the parameters involved. 


\bigskip

\paragraph{\bf Existence of normal and vortex lattice solutions.} We say an operator $A$ is magnetically translation invariant ({\it mt-invariant}, for short), if 
  it  satisfies $\tau_{b h} A = A,\ \forall h\in \R^2$. 

 The existence of the mt-invariant normal states for $b\ne 0$ is stated in the following theorem proven in Section \ref{sec:norm-states-exist}: 

\begin{theorem}\label{thm:norm-state-exist} 
Drop the exchange term $v^\sharp \gamma$ and let 
$v$ be $\lat$-periodic and either $\int v>0$, 
 or $\|v\|_\infty<\infty$. Then the BdG equations \eqref{Gam-eq}-\eqref{a-eq} 
 on the  space $I^{2,1}\times I^{2,2}\times \frachbvec^{2}$  have  
 a mt-invariant  solution, unique on the set of even (in the sense of the definition \eqref{even}) pairs $(\G, a)$, with $\G$ mt-invariant (i.e. satisfying \eqref{mt-invar}). 
 
Moreover, this solution is normal (i.e. $\al=0$) and is of the form  $(\G_{T, b},\  a_b)$ and 
\begin{align}\label{Gam-Tb} 
		\G_{T b }:= \left( \begin{array}{cc} \g_{T b } & 0 \\ 
		0 & \one-\bar\g_{T b } 	\end{array} \right), 
\end{align}
with $\g_{T b }$  solving the equation 
\begin{align}\label{gam-Tb}\gamma= f_{T}(h_{\g a_b }),\end{align}
where $f_{T}(\lam)$ 
  is given in \eqref{FD-distr}. 
\end{theorem}
The search for  mt-invariant solutions is simplified by the following statement proven in Section \ref{sec:norm-states-exist}:  
\begin{proposition}\label{prop:mt-inv-al0}
  If $\G$ is mt-invariant, then  $\G$ is normal, i.e. $\alpha = 0$.
\end{proposition}

Recall that $b=\frac{2\pi n}{|\lat|}$, see \eqref{b-quant}. For vortex lattices, we have the following result
 
\begin{theorem} \label{thm:BdGExistence}
 Drop the 
the exchange term  $v^\star \g=v* \rho_\gamma$ in the definition of $h_{\g a}$ in \eqref{Lam}. 
Fix a lattice $\lat$ and a value of the Chern number $c_1(\chi^\lat) = n\in \Z$ (see 
\eqref{mf-quant}) 
and assume that $v$ is $\lat$-periodic and obeys 
 $\|v\|_\infty< \infty$. 
 Then 
 
(i)  for any $T \ge 0$,  
 there exists a (weak)  solution $(\G, a) \in \mathcal{D}^1_\nu \times \frachvec^1$ 
of the BdG equations \eqref{Gam-eq}-\eqref{a-eq} (in particular, it satisfies $T^{\rm trans}_{ s} (\G, a) =\hat T^{\rm gauge}_{\chi_{s}^\lat} (\G, a)$), which  minimizes the free energy $F_T$ 
  (for the  given $c_1(\chi^\lat) = n\in \Z$); 


(ii) 
for  $T$ and $b$  sufficiently small, 
$(\G,a)$ has $\alpha \not= 0$, i.e. this solution is a vortex lattice, provided 
$v(x) \ls - (1+|x|)^{-\kappa}, \kappa<2$,  for $x$  in a fundamental cell $\Om$ centred at the origin. 
 More generally, the latter holds if  the operator $L_{Tb}$, given in \eqref{LTb}  and defined on $ I^{2, 2}$,
has a negative eigenvalue.
\end{theorem}
Statement (i) is proven in Section \ref{sec:vort-latt-exist} and (ii) follows from Proposition \ref{thm:T-b-small}. 

\bigskip
\paragraph{\bf Stability/instability of normal solutions.}

We address the question of the energetic stability of the mt-invariant states. To this end we define, in the standard way, the Hessian of the free energy in $\G$ as
\begin{align}\label{Hess-eta}F_T''(\G_{*}, a_*):=d_\G \grad_\G F_T(\G_{*},a_*),\end{align}
where $d_\G $ is the G\^ateaux derivative w.r.to $\G$ and $\grad_\G$ is the gradient w.r.to $\G$, defined by the equation \[\tr( \grad_\G F(\G_{*},a_*) \G')= d_\G F(\G_{*},a_*) \G'.\]

We  consider $F_T''(\G_{T b }, a_b)$ 
along physically relevant perturbations 
 of the form $\G' = \phi(\alpha)$, where $ \phi(\alpha)$ is 
the off-diagonal operator-matrix, defined by 
\begin{align}\label{Phi-al}\phi(\al) :=\left( \begin{array}{cc}
										0 & \al \\
										\al^* & 0
						\end{array} \right) ,
\end{align}
with $\al$  a Hilbert-Schmidt operator on $\frach$. 
Moreover, we require that perturbations  $\G'= \phi(\alpha)$ satisfy the condition 
\begin{align}\label{al-cond}
	\alpha \alpha^* \lesssim [\gamma_{Tb}(1-\gamma_{Tb})]^2
\end{align}
(which is equivalent to  \eqref{Gam'-cond} at $\G=\G_{T b} $, c.f. also Lemma \ref{lem:al-gam-bnd}(1)). 

 For an operator $h$, let $h^L$ and $h^R$ stand for the operators acting on other operators by multiplication  from the left by $h$  and from right by $\bar h$, respectively, and recall $v^\sharp$ is defined after \eqref{BdG-eq-t}. 
We have the following
 \begin{proposition}\label{prop:FT''}
For off-diagonal perturbations $\G' = \phi(\alpha)$, $F_T''(\G_{Tb}, a_b) \phi(\alpha)= \phi(L_{Tb}\alpha)$, where the operator  $L_{Tb}$ 
  is given by 
\begin{align}\label{LTb}
	& L_{Tb} := K_{Tb} + v^\sharp,\\
\label{KTb}	& K_{Tb} := \frac{h_{T b}^L+ h_{T b}^R}{\tanh(h_{T b}^L/T)+\tanh( h_{T b}^R/T)}, 
\end{align} 
 on the space  of Hilbert-Schmidt operators on $\frach$. Here $h_{T b}:=h_{\gamma_{Tb} \ab \mu}$ (with $h_{\g a \mu}:=h_{\g a}-\mu$, where $\hga$ is given in \eqref{Lam}). 
   \end{proposition} 
 Let $\langle \alpha,  \alpha' \rangle:= \Tr(\al^*\al')$. We say that $\G'$ is an {\it off-diagonal perturbation} if and only if  $\G' = \phi(\alpha)$ with $\al$ a Hilbert-Schmidt operator satisfying  \eqref{al-cond}. 
The next result 
  generalizes that of \cite{HHSS} for $a=0$: 

   \begin{proposition}\label{prop:FT-expan-order2}  
For off-diagonal perturbations $\G' = \phi(\alpha)$, 
 we have
\begin{align}\label{cF-expan2ord}
	F_T(\G_{Tb}+ \e\G', a_b) =& F_T (\G_{Tb},a_b) + \e^2\langle \alpha, L_{Tb} \alpha \rangle + O(\e^3).
\end{align}
 \end{proposition}

 Definition  \eqref{energy} of $E(\G, a)$ implies that $E(\G_{Tb}+ \e\G', a_b) = E(\G_{Tb}, a_b) + \e^2\Tr(\bar{\alpha} v^\sharp \alpha)$ for $\G'=\phi(\alpha)$. 
  This together with Corollary \ref{cor:S'-S''} and \eqref{S-expan} of Appendix \ref{sec:entropy} dealing with the entropy 
   yields Propositions \ref{prop:FT''} and \ref{prop:FT-expan-order2}, respectively.
\DETAILS{Now, proceed to the proof of Proposition \ref{prop:FT-expan-order2}. Using the definition \eqref{FT-def}, with \eqref{energy}, and the 
expansion  \eqref{S-expan}, we find, for perturbations $\G' = \phi(\alpha)$ obeying the condition \eqref{Gam'-cond}, 
\begin{align}
	F_T(\G_{Tb}+ \e\G', a_b) =& F_T(\G_{Tb}, a_b) + \e^2\Tr(\bar{\alpha} v^\sharp \alpha) +  \e^2 S''(\G',\G') + O(\epsilon^3). 
\end{align}}
(The absence of the linear term is due to the fact that $(\G_{Tb},0)$ is a critical point of the energy function.) 

The next two propositions are proven in Section \ref{sec:normal-stab/instab}.
  \begin{proposition}\label{prop:T-b-large}
 $L_{Tb}\ge  \frac12 T-\|v\|_\infty$  and consequently,  for $T$ sufficiently large, $L_{Tb}>0$ and the   mt-invariant state $(\G_{T, b},  a_b)$ is energetically stable under off-diagonal perturbations. 
   \end{proposition} 

On the other hand,  for 
 $T$ and $b$ sufficiently small, we have
   \begin{proposition}\label{thm:T-b-small}
Drop the 
exchange term $-v^\sharp \gamma$.
Suppose that  $v(x) \ls - (1+|x|)^{-\kappa}, \kappa<2$,  for $x$  in a fundamental cell $\Om$ centred at the origin. 
 Then, for  $T$ and $b$ sufficiently small and $ T\ls b^\sigma, \sigma >\kappa/2$,  the operator $L_{Tb}$ acting on the set of Hilbert-Schmidt operators on $\frach$ has a negative eigenvalue and consequently 
  the  mt-invariant state $(\G_{T, b},  a_b)$ is energetically unstable under 
  perturbations $\G'=\phi(\alpha)$. 
   \end{proposition}
 Note that $L_{Tb}$ is discontinuous at $T=0$. 

 Let $T_c(b)$ (resp. $T_c'(b)$) be the largest (resp. smallest) temperature s.t.  the normal solution is energetically unstable (resp. stable) under off-diagonal perturbations $\G' = \phi(\alpha)$, for $T<T_c(b)$ (resp. $T>T_c(b)')$. Clearly, $\infty \ge T_c'(b)\ge T_c(b)\ge 0$. 
Propositions \ref{prop:T-b-large} and  \ref{thm:T-b-small} imply 
\begin{corollary}\label{cor:Tcb-posit} 
Under the conditions of  Proposition \ref{thm:T-b-small}, $T_c(b)>0$ for $b $ sufficiently small and  $T_c'(b) =0$ for $b $ sufficiently large. \end{corollary}
We {\it conjecture} that  (a) mt-invariant  states coincide with normal states, (b) For $T\gg 1$, the free energy $F_T$ is minimized by normal states; (c)  normal states are stable if either  $T\gg 1$ or  $b\gg 1$ and unstable  if  $T\ll 1$ and  $b\ll 1$; (d)  $T_c(b)= T_c'(b) .$

The next corollary provides a convenient criterion for the determination of $T_c(b)$ and $T_c(b)' $.
 \begin{corollary}\label{cor:Tc-d F_T'ev0}
At $T=T_c(b)$ and $T=T_c'(b) $, zero is the lowest eigenvalue of the operator  $L_{Tb}$. 
   \end{corollary}
 A proof of energetic stability under general perturbations for either $T$ or $b$ sufficiently large is more subtle. 
For it, 
 one has to use the full linearized operator, $\hess F_T(\G_{Tb},a_b)$. 
Our computations 
 suggest that $0$ is the lowest eigenvalue of $\hess F_T(\G_{Tb},a_b)$ if and only if $0$ is the lowest eigenvalue of $L_{Tb}$ and, consequently, $T_c(b)$ and $T_c'(b) $ apply also to the general perturbations. 

The statement that $T_c=T_c(0)=T_c' (0)>0$ for $a=0$ (and therefore $b=0$) and for a large class of potentials is proven, 
 by  the variational techniques, in \cite{HHSS}. 
 
In conclusion of this paragraph, we mention the proposition which follows from a simple computation: 
  \begin{proposition}\label{prop:LTb-comm-MT}
The operator  $L_{Tb}$ commutes with the magnetic translations. The same is true for the $\G-$Hessian $F_T''(\G_{Tb},a_b)$ (see \eqref{Hess-eta}).  \end{proposition}

\paragraph{\bf Remarks.} Here we collect the comments on various aspects covered in this introduction.


\begin{rem}[BdG equations]\label{rem:space} 
{\em 
1) For $\s=0$, \eqref{BdG-eq-t} 
 becomes 
 the  time-dependent von Neumann-Hartree-Fock equation for $\g$ (and $a$).

2) 
A formal derivation of \eqref{BdG-eq-t} given in Appendix \ref{sec:qf-reduct} can be extended to the coupled system \eqref{BdG-eq-t}-\eqref{Amp-Maxw-eq} by starting with the many-body quantum Hamiltonian \eqref{H} coupled to the quantized electro-magnetic field. 

}\end{rem}

\begin{rem}\label{rem:BdG2} 
{\em 
\DETAILS{ 2)  For \eqref{Gam-eq'} to give $\G$ of the form \eqref{Gam}, the function  $f (h)$ should satisfy the conditions
\begin{equation} \label{g-cond}
	f(\bar h) = \overline{f(h)} \text{   and   } f(-h) =\one - f(h) .
\end{equation}
 For  $g(x)$ given in \eqref{g}, the function  $f(h):= (g')^{-1}(h)$ satisfies these conditions as can be checked from its explicit form  \eqref{FD-distr}.  
However, \eqref{g-cond} is more general than \eqref{FD-distr}. 
Indeed, the first condition in \eqref{g-cond} means merely that $f$ is a real function, 
  while the second condition in \eqref{g-cond} is satisfied by functions $f(h):= (1+e^{\tilde{g}(h)})^{-1}$, with  $\tilde{g}(h)$, any odd function. 
  \DETAILS{To check this we set $g(h) = f(h)^{-1}-1$. Then we see that it is equivalent to
\begin{align}
	g(h)g(-h) =1 .
\end{align}
It is easy to see now that $f(h):= (1+e^{\tilde{g}(h)})^{-1}$, with  $\tilde{g}(h)$ odd, satisfies \eqref{g-cond}.}

 More generally, one could require that  $g(\lam)$ satisfies the conditions \eqref{g-cond}, with $f(h):= (g')^{-1}(h)$, and 
\begin{align}\label{g-cond-2}   g'(1-x)=-g'(x).\end{align}}
It is common in physics literature to drop the direct and exchange  self-interaction terms 
 from $h_{\g a \mu}$ in \eqref{Lam'}. In this case, $\Lam_{\G a }$ becomes independent of the diagonal part, $\g$, of $\G$ and equation 
  \eqref{Gam-eq} has always the solution}
\begin{equation}\label{Gam-sol-diag}
	\G_{T a }\ =f_{T}(\Lam_{a }),\ \text{ where }\   \Lam_{a }:= \Lam_{\G a }\big|_{\G=0}.
\end{equation}

Similarly, for equation \eqref{gam-Tb} for normal states. In this case, it always has the solution
\begin{align}\label{gam-Tb'}\gamma_{T b }= f_{T}(h_{ a_b }),\ \text{ where }\    h_{ a}  =-\Delta_a-\mu.\end{align}\end{rem}
 


\begin{rem}[Particle density]\label{rem:den} {\em For a definition of $\rho_\g$ not relying on the integral kernels, see \eqref{den-def}.
}\end{rem}


\begin{rem}[Particle-hole symmetry] \label{rem:part-hole-sym}   The evolution under the BdG equations \eqref{BdG-eq-t}-\eqref{Amp-Maxw-eq} preserves the relations in \eqref{Gam-prop}, i.e. if an initial condition has one of these properties, then so does the solution.  
This follows from the relation 
\begin{align} \label{Lam-ph-sym} J^*\Lam J= -  \overline{\Lam},\end{align}
where $J$ is defined in \eqref{Gam-prop}.  The second relation in \eqref{Gam-prop} is called the {\it particle-hole symmetry}.
\end{rem}

\begin{rem}[BdG equations in `coordinate' form]\label{rem:BdG-EVP}   {\em  
In physics literature, the BdG equations are written 
  in terms of eigenfunctions of the operator $\G$, or $\Lam_{\G a}$ (assuming the operator $-\Delta_a$ has purely discrete spectrum, see \cite{BLS, Cyr, dG, Kop, Tink, J-X-Zhu}, cf. \cite{BBCFS}). In the case of stationary equations \eqref{Gam-eq}-\eqref{a-eq}, we have  
\begin{align}\label{Lam-EVP}   &\Lam_{\G a}\psi_n = \veps_n\psi_n,\ \quad \Lam_{\G a}J\bar\psi_n = -\veps_nJ\bar\psi_n,\end{align}
with $\veps_n>0$ and $J$ given in \eqref{Gam-prop}. The second equation above follows from the first and relation \eqref{Lam-ph-sym}. By \eqref{Gam-eq'}, $\psi_n$ and $J\bar\psi_n$ are also eigenfunctions of $\G$ with the eigenvalues $f_n:=f_T(\veps_n)$ and $1-f_n=f_T(-\veps_n)$, respectively.  As $\{\psi_n, J\bar\psi_n\}$ form a complete orthonormal set, we have 
\begin{align}\label{eta-exp}   &\G=\sum_n \big(f_n P_{\psi_n} + (1-f_n) P_{J\bar\psi_n}\big),\end{align}
where $P_{\psi}$ stands for the rank-one orthogonal projection onto the subspace spanned by $\psi$.  \eqref{Lam-EVP}-\eqref{eta-exp}  yield a nonlinear eigenvalue problem replacing   \eqref{Gam-eq}. Finally, one can express $j(\G, a)$ in \eqref{a-eq}  in terms of $\{\psi_n, J\bar\psi_n\}$.
}\end{rem}

\begin{rem}[Entropy]\label{rem:entr} 
{\em   Due to the symmetry \eqref{Gam-prop} of $\G$, we see that 
\begin{align}\label{S-sym}
	\Tr(\G \ln \G) = \Tr((1-\G)\ln(1-\G))
\end{align}
which, recalling \eqref{g}, implies that} 
\DETAILS{$\tr g(\G)=\tr s(\G)$, where 
\begin{align}
	s(\G):=-2 \G\ln\G \label{eqn:s-def}
\end{align}
 and  $g(\lam)$ is given in \eqref{g}, and therefore 
 \begin{align}\label{S-expr} S(\G) =\tr s(\G) = \Tr g(\G).\end{align}}
\begin{align}
\label{S-expr}	& S(\G) := \Tr(s(\G))= \Tr(g(\G)),\\
 \label{g-s-def}	&  g(\G): = -\G\ln\G - (1-\G)\ln(1-\G),\ s(\G ):= -2\G\ln\G.
\end{align}\end{rem}

 \DETAILS{ Since $T_s$ is a group representation, $\chi_{s} (x)$ must satisfy
\begin{align}\label{chi-co-cycle} 
 \chi_{s+t} (x) - \chi_s(x)- \chi_t (x+s)  \in 2\pi\Z, 
\end{align}
for any $s, t \in \R^d$. 
{\bf In Subsection \ref, we will present a solution to this equation in dimension $2$ of the interest for us.}}
\begin{rem}[Co-cycle equation]\label{rem:co-cycle}  
   {\em One can readily show that the operator-family $u_{b s}, s\in \lat,$ defined in \eqref{ubs} is a group representation (with the group law given in \eqref{ubs-gr}) if and only if the functions $\chis(x), s\in \lat, x\in \R^2,$ 
  satisfy the co-cycle condition 
\begin{align}\label{co-cycle}
    \chi_{s+t}^\lat(x) -\chi_s^\lat (x+t) - \chi_t^\lat(x) \in 2\pi \Z,\ \forall s, t\in \lat.
\end{align}
The function $\chis$   defined in \eqref{chibs} satisfies this relation, 
while the one in \eqref{chibs'} does not. 

Functions  $\chi_\cdot : \lat\times \R^d\ra \R$ satisfying \eqref{co-cycle} are called the {\it summands of automorphy}, or {\it co-cycles}  (see \cite{Sig0} for a relevant discussion). (The map 
$e^{i \chi} : \lat\times\R^2 \to U(1)$, where 
 $\chi (x, s) \equiv \chi_s (x)$ is called the {\it factor of automorphy}.)
Every map $\chi_s : \lat\times\R^2 \to \R$ satisfying \eqref{co-cycle} is (gauge) equivalent to  the one in \eqref{chibs}, with  $b$ satisfying \eqref{b-quant}.

With a summand of automorphy $\chis(x), s\in \lat, x\in \R^2,$  one associates the function 
 \begin{equation}\label{c1}
    c_1(\chi) =\frac{1}{2\pi} (\chi_{\nu_2}(x+\nu_1) - \chi_{\nu_2}(x) - \chi_{\nu_1}(x+\nu_2) + \chi_{\nu_1}(x)), 
\end{equation}
  for some basis $\{\nu_1, \nu_2\}$ in $\cL$.   As it turns out, this function is independent of $x$ and the choice of the basis and is an integer. In physics literature, $c_1(\chi)$ is called the Chern number.} 
\DETAILS{{\bf It is straightforward to show that $\chi^\lat=\chis(x)$, defined in \eqref{chibs}, is  a summand of automorphy 
 and $c_1(\chi^\lat)$ could be given by \eqref{c1}, 
  without reference to the explicit form \eqref{chibs} of $\chi^\lat$.}} 
\end{rem}

\begin{rem}[Operator $L_{Tb}$]\label{rem:LTb}  {\em 1)  The question of when $L_{Tb}$ has negative spectrum for a larger range of $T$'s is a delicate one. For $T$ close to $T_c$, this depends, besides of the parameters $T$ and $b$, also on $v$ and $\mu$ determining whether the superconductor is of Type I or II. (As was discovered theoretically by A.A. Abrikosov in his study of the Ginzburg-Landau equations (\cite{Abr}) and confirmed later experimentally, superconductors are divided, according to their basic properties,  into two groups, Type I or II superconductors. 
So far, there are no results (rigorous, or not) establishing within the BdG theory these properties for certain classes of potentials $v$ and $\mu$.) 

2) 
Since the components of magnetic translations \eqref{mag-transl} do not commute, the fiber decomposition of $L_{Tb}$ is somewhat subtle (see \cite{AHS}).} 
\end{rem}





\paragraph{\bf Nomenclature.} 
 In our most important results, we drop the exchange term $v^\sharp \gamma$.  By analogy with the Hartree-Fock equations, the resulting equations could be called the {\it reduced BdG equations}. 

  An important and natural modification of the BdG equations would be replacing  the exchange term $v^\sharp \gamma$ by a density dependent exchange-correlation term $xc(\rho_\g)$ from the density functional theory. The resulting equations could be called the {\it density  BdG equations}.

Addressing the density or original BdG equations  would be an important next step. 

\smallskip

 The paper is organized as follows. In Sections \ref{sec:norm-states-exist} and \ref{sec:vort-latt-exist}, we prove Theorems \ref{thm:norm-state-exist} and \ref{thm:BdGExistence}, on existence of the normal and vortex lattice solutions, respectively. These are our principal results. In Section \ref{sec:normal-stab/instab}, we prove Propositions  \ref{prop:T-b-large} and \ref{thm:T-b-small} on the stability/instability of the normal solutions.  
 In Appendix \ref{sec:entropy} we study the entropy functional. Results of this appendix imply the proofs of Propositions \ref{prop:FT''} and \ref{prop:FT-expan-order2} and are used  in Appendix \ref{sec:energy} in the proof of technical Theorem \ref{thm:BdG=EL}. 
  In 
  Appendix \ref{sec:xi-fp}, we prove an elementary technical result from Section \ref{sec:norm-states-exist}. Finally,  in Appendix \ref{sec:den-est} we prove some bounds on functions relative to magnetic Laplacian and bounds on densities (both elementary and probably well-known)  and in Appendix \ref{sec:qf-reduct}, we discuss the derivation of the BdG equations from the quantum many-body problem.


 \section{Stationary Bogoliubov-de Gennes equations 
  }\label{sec:BdG-stat} 
 In this section we establish the connection between the time-dependent and stationary BdG equations and show that the latter are the Euler-Lagrange equations for the free energy, \eqref{FT-def}.
 
In terms of $\G$, the gauge transformation, $T^{\rm gauge}_\chi$, could be written as 
\begin{equation}\label{gauge-transf'} 
  T^{\rm gauge}_\chi : 	\G\ra U^{\rm gauge}_\chi \G (U^{\rm gauge}_{\chi})^{-1},\ \text{ where }\ U^{\rm gauge}_\chi=\left( \begin{array}{cc} e^{i\chi }  & 0 \\ 0  & e^{-i\chi }  \end{array} \right).
\end{equation}
It is extended correspondingly to $(\G, a)$ by $T_\chi^{\rm gauge} (\G , a)=(T_\chi^{\rm gauge} \G , a+\nabla \chi)$. Notice the difference in the action 
  on the diagonal and off-diagonal elements of $\G$.

The  invariance under the gauge transformations follows from the relation
\begin{equation}\label{Lam-gauge-cov}
	\Lambda(T_\chi^{\rm gauge} (\G , a))=  T_\chi^{\rm gauge} (\Lambda(\G, a)),
\end{equation}	
 shown by using the operator calculus.

\smallskip



We consider stationary solutions to \eqref{BdG-eq-t}, with $a$ independent of $t$, of the form 
\begin{align}\label{stat-sol}
   \G_t :=  T^{\rm gauge}_\chi \G = U_\chi^{\rm gauge} \G  (U_\chi^{\rm gauge})^{-1},
\end{align}
where $\G$ 
 is independent of $t$.  
 We have  \begin{proposition}\label{prop:stat-sol} 
The operator-family \eqref{stat-sol}, with $\G$ independent of $t$ and  $\dot\chi$, constant, say $\dot\chi\equiv -\mu$,  is a solution to \eqref{BdG-eq-t}, if and only if $\G$ solves the equation 
\begin{equation}\label{BdG-eq-stat'}
[\Lam_{\G a}, \G ]=0, 
\end{equation}
where $\Lam_{\G a}\equiv \Lam_{\G a \mu}$ is given explicitly  in \eqref{Lam'}.
\DETAILS{\begin{align} \label{Lam'} 
 \Lam_{\G a }:=\left( \begin{array}{cc} h_{\g a \mu} & v^\sharp \s \\ v^\sharp \s^* & - \bar h_{\g a \mu} \end{array} \right), \end{align}
with $h_{\g a \mu}:=h_{\g a}-\mu$ and, recall, $\hga  =-\Delta_a+ 
  v^* \gamma  - v^\sharp\, \gamma$ (see \eqref{hgam-a}).}
\end{proposition}
\begin{proof}   
We write $U_\chi^{\rm gauge}\equiv U_\chi$ and use that $\partial_t U_\chi = i \dot \chi S U_\chi $, where
\[S:=\left( \begin{array}{cc} 1 & 0 \\ 0 & - 1 \end{array} \right),\] and therefore $\partial_t (U_\chi \eta U_\chi^{-1}) = i\dot \chi [S, U_\chi \eta U_\chi^{-1}]= i\dot \chi U_\chi [S, \eta] U_\chi^{-1}$. Plugging \eqref{stat-sol} into \eqref{BdG-eq-t},  and using \eqref{Lam-gauge-cov}, the previous relation  and that $\chi$ is independent of $x$, we see that 
\begin{align}
	-\dot \chi U_\chi [S, \eta] U_\chi^{-1} 
	 =& [\Lambda(U_\chi \eta U_\chi^{-1}, a), U_\chi \eta U_\chi^{-1} ] \\
		=& [\Lambda(U_\chi \eta U_\chi^{-1}, a+\nabla \chi), U_\chi \eta U_\chi^{-1} ] \\
		=& U_\chi [\Lambda(\eta, a),\eta] U_\chi^{-1}.
\end{align}
 Since $\dot\chi\equiv -\mu$, 
 it follows then that
$	[\Lambda(\eta, a)-\mu S,\eta]=0.$ Since $
 \Lambda(\G, a) - \mu S$ is equal to \eqref{Lam'}, the last equation is exactly \eqref{BdG-eq-stat'} which gives the statement of the proposition.
\end{proof}

 For any reasonable function $f$ and  time-independent $a$, solutions of the equation 
 \begin{align} \label{Gam-eq''} 
	\G=f(\frac{1}{T} \Lambda_{\G a}),
\end{align} 
solve \eqref{BdG-eq-stat'} and therefore give stationary solutions of \eqref{BdG-eq-t}. 
Conversely, solutions to \eqref{BdG-eq-stat'}, s.t. the spectrum of $ \Lambda_{\G a}$ is simple, solve \eqref{Gam-eq''}. (The parameter $T>0$, the temperature, is introduced here for the future reference.)

Inverting the function $f$, one can rewrite \eqref{Gam-eq''} as $\Lambda_{\G a}=T f^{-1}(\G)$. Let $f^{-1}=: g'$. Then the stationary BdG equations 
can be written (in the Coulomb gauge $\divv a=0$) as \eqref{Gam-eq}-\eqref{a-eq}.
\DETAILS{\begin{align} \label{Gam-eq}
	&\Lam_{\G a} -T g'(\G) =0,\\ 
 \label{a-eq} 
     &\CURL^*\CURL a - j(\G, a)=0.
\end{align}}

The physical function $f$ is selected by either a thermodynamic limit (Gibbs states) or by coupling the system in question to a thermal reservoir (or imposing  the maximum entropy principle). 
 It is given by the Fermi-Dirac distribution  \eqref{FD-distr}
It follows from the equations $g'=f^{-1}$ and \eqref{FD-distr} that   
the function $g'$ is given by \eqref{g'} and the function $g$ is equal to 
\begin{equation} \label{g}
	g(\lam)= - \lam\ln \lam  - (1-\lam) \ln (1-\lam).
\end{equation}
From now on, 
 {\it we assume  that $f(\lam)$ and $g(\lam)$ are given in \eqref{FD-distr} and \eqref{g}, respectively}. 

\medskip


Our next result shows that the BdG equations arise as the Euler-Lagrange equations for the free energy functional   $F_{T}$ differentiated 
 along the perturbations (`tangent vectors') at $(\G, a)$ from the following class (c.f. also Lemma \ref{lem:al-gam-bnd}(1))
\begin{align}\label{pert-class}& \cP_\G:=\{(\G', a') \in \hat I^1 \times \vec H^1:  \text{ \eqref{Gam'-cond} holds } \},\\
\label{Gam'-cond}  
&J^* \G' J=-\bar\G',\   (\G')^2 \lesssim [\G (1-\G)]^2, \text{ and } \Tr(S_1\G') = 0,
\end{align}
where $S_1 = \text{diag}(1,0)$ and $J$ is defined in \eqref{Gam-prop}. 
 Conditions \eqref{Gam'-cond} 
  are designed to handle a delicate problem of non-differentiability of $s(\lam):=2 \lam \ln \lam$ at $\lam=0$, while providing a sufficiently rich set, $\cP_\G$, to derive the BdG equations. 

\begin{theorem} \label{thm:BdG=EL}
(a) The free energy functional   $F_{T}$ is well defined on the space $ \mathcal{D}^1_\nu \times \frachvec^1$.

(b) 
$F_T$ is continuously (G\^ateaux  or Fr\'echet) differentiable at $(\G, a) \in \mathcal{D}^1_\nu \times \frachvec^1$, 
   with respect of perturbations $(\G', a')\in \cP_\G$. 

(c) 
Critical points of $F_T$, which are even (in the sense of the definition \eqref{even}) and obey  $0 < \G < 1$, satisfy 
 the  BdG equations for some $\mu$ (determined by the constraint $\Tr \g = \nu$). 
 
(d)  Minimizers of $F_T$ over $\mathcal{D}^1_\nu \times \frachbvec^1$ satisfy the conditions of statement (c) and  therefore satisfy the  BdG equations for some $\mu$. 
\end{theorem}

This theorem is proven in  Appendix \ref{sec:energy}. For the translation invariant case, it is proven 
  in \cite{HHSS}. 
For $a=0$,  but $\G$ not necessarily translation invariant,  the fact that the BdG equations are the Euler-Lagrange equations of the free energy functional is used in  \cite{FHSS}, but, it seems, with no proof provided. 

As a result of Theorem \ref{thm:BdG=EL}, we can write the  BdG equations as 
\begin{align}\label{F'-eq} 
F'_{T }(\G, a)=0, \end{align}
where the map  $F'_{T }(\G, a)$ is defined by the r.h.s. of the static BdG equations \eqref{Gam-eq}-\eqref{a-eq} and can be thought of as a gradient map of $F_T$.
 \DETAILS{Hence,   we can write the equations  \eqref{Gam-eq}--\eqref{a-eq} 
as 
\begin{align}\label{F'-eq} 
F'_{T }(\G, a)=0,\end{align}
 where $F'_{T }(\G, a)=( F'_{T \eta}(\G, a), F'_{T a}(\G, a))$ 
  is the formal G\^ateaux derivative (more precisely, the trace metric gradients) of  
  the free energy,  $ F_{T}(\G, a)$, in $\G$ and $a$.} 
\DETAILS{\begin{align}\label{F'T} 
	F'_{T }(\G, a):=\Lam_{\G a } -T g'(\G). 
\end{align} } 

\medskip


\section{The normal states: 
 Proof of Theorem \ref{thm:norm-state-exist}}\label{sec:norm-states-exist} 
 
We begin with some results about the magnetic translations.
Recall  the operators $T_{b s}$ and $U^{\rm gauge}_{\chi}$, defined in \eqref{mag-transl} and \eqref{gauge-transf'}. Group properties of  
   $T_{b s}$ are established in following 
 \begin{lemma} \label{lem:mtProjRep'}
The operators $T_{b s}$, defined in \eqref{mag-transl} and restricted to $\G$'s, 
 satisfy 
\begin{align} \label{TsTt}	& T_{b s}T_{b t} = \hat I_{b s t}
  T_{b s+t},\end{align}
  where $\hat I_{b s t}$ acts on bounded operators on $\mathfrak{h}$ as 
  \begin{align}  \label{Ist}&\hat I_{b s t} A:= I_{b s t} A I_{b s t}^{-1},\ I_{b s t} := 
  U^{\rm gauge}_{\frac b2 (s \wedge t)}.
\end{align}
\end{lemma}

\begin{proof}[Proof of Lemma \ref{lem:mtProjRep'}]
Let $U^{\rm trans}_{ s}:=\diag(u_s^{\rm trans}, u_s^{\rm trans})$  on $L^2_{\rm loc}(\R^2)\times L^2_{\rm loc}(\R^2)$ and define the transformations 
 \begin{align}\label{Ubs}U_{b s}:= (U^{\rm gauge}_{\chib})^{-1}U^{\rm trans}_{ s}.\end{align} Then 
 $T_{b s} \G =U_{b s} \G U_{b s}^{-1}$. Using the definition above and the relation $U^{\rm trans}_{ s}U^{\rm gauge}_{\chi}=U^{\rm gauge}_{u_s^{\rm trans}\chi} U^{\rm trans}_{ s}$, we compute 
 \[U_{b s} U_{b t}= (U^{\rm gauge}_{\chib+ u_s^{\rm trans}\chibt})^{-1}U^{\rm trans}_{ s+t}= (U^{\rm gauge}_{\chib+ u_s^{\rm trans}\chibt-\chi^b_{s+t}})^{-1}U_{b s+t}.\] By definition \eqref{chibs'}, 
 we have $\chib+ u_s^{\rm trans}\chibt-\chi^b_{s+t} =- \frac b2(s\wedge t)$, 
 which implies $U_{b s} U_{b t}=I_{b s t} 
 U_{b s+t}=  U_{b s+t}I_{b s t} $.  
  Hence, the result follows.
\end{proof}

Now,  Lemma \ref{lem:mtProjRep'} implies Proposition \ref{prop:mt-inv-al0}:
 \begin{proof}[Proof of Proposition \ref{prop:mt-inv-al0}]
 \eqref{TsTt}-\eqref{Ist} and the mt-invariance, \eqref{mt-invar}, yield $\G=T_{b s}T_{b t} \G=I_{b s t} \G I_{b s t}^{-1}$, which  implies that $\alpha 
= e^{i\frac b2 (s \wedge t)}\alpha$ for all $s,t \in \R^2$. This yields that $\alpha = 0$.
\end{proof}
Note that, {\it ignoring} $a$, the fact that $T_{b s}$ maps the set of bounded diagonal operator-matrices on $\frach\times \frach$ into itself  follows from  \eqref{TsTt}-\eqref{Ist} and the fact that $\hat I_{b s t}$ is an identity on diagonal operator-matrices.

  Proposition \ref{prop:mt-inv-al0} shows that we can restrict our search of  mt-invariant solutions to normal states.
For normal states, i.e. for $\alpha = 0$, the BdG equations \eqref{Gam-eq}-\eqref{a-eq} (with  \eqref{Gam-eq} replaced by \eqref{Gam-eq'})  reduce to the following equations for $\g$ and $a$
\begin{align}
	& \gamma = f_{T}( h_{\g a  \mu}),  \label{eqn:NormalBdGPart1} \\
	& \curl^* \curl a = j(\g,a), \label{eqn:NormalBdGPart2}
\end{align}
where, recall, $j(\g, a)(x) :=([-i \na,\g]_+)(x, x)$. These are the coupled Hartree-Fock and Amp\'ere equations.

 Recall that $\ab(x)$ is the magnetic potential with the constant magnetic field $b$ ($\curl a_b=b$) in the gauge s.t. 
  $\ab(x)=\frac b2 *x, *x:=(-x_2, x_1)$. First, we show that the second equation is automatically satisfied for $a=a_b$ and $\gamma$ a magnetically translation invariant operator, which is even in the sense of \eqref{even}.

 Define  the magnetic translations $u_{b s}$ as (cf. \eqref{ubs} and \eqref{Ubs})
\begin{align}\label{ubs'} u_{b s} := u^{\rm gauge}_{-\chib} u^{\rm trans}_{ s},\end{align} where, recall the operators $u^{\rm gauge}_{\chi}$ and $u^{\rm trans}_{ s}$ are defined are defined after  \eqref{ubs} 
  and $\chi_{ \cdot}^b: \lat\times \R^d\ra \R$ is given in \eqref{chibs'}. Recall  the definition of the space $\frach$ in \eqref{fh-lat}.
      \begin{lemma}\label{lem:grad-den} With  the definition 
  $\tau_{b s}(\g)=u_{b s}\g {u_{b s}}^{-1}$, we have (cf. \eqref{TsTt}) 
\begin{align} \label{usuh'}	& u_{b s} u_{b h} = e^{-i b h\wedge s}  u_{b h}u_{b s},\\
\label{TsTt'}	  &\tau_{b s} \tau_{b h}=\tau_{b h}\tau_{b s},\\   
&\label{frachb-invar} \text{If $\g$ maps $\frach$ into itself, then so does } \tau_{b h}\g.\end{align} 
\end{lemma}
\begin{proof} 

 \eqref{usuh'} follows from $u_{b s} u_{b h} = e^{-i\frac b2 h\wedge s} e^{-i\chi_{b s}}  u_{b h}\tau_s^{\rm trans}= e^{-i \frac b2 h\wedge s} e^{i\frac b2 s\wedge h}u_{b h}u_{b s}$. 
 To prove \eqref{TsTt'}, we 
 use that $\tau_s^{\rm trans} e^{-i\chi_{b h}}=e^{-i\frac b2 h\wedge s} e^{-i\chi_{b s}} \tau_h^{\rm trans}$, where, recall, $\chi_{b s}$ is defined in \eqref{chibs}, 
and \eqref{usuh'} which yield $(u_{b s} u_{b h})^{-1}=e^{i b h\wedge s} (u_{b h}u_{b s})^{-1}$ and therefore, due to  $\tau_{b h}\g:=u_{b h}\g u_{b h}^{-1}$, \eqref{TsTt'} follows.

Finally, $\g$ maps the space $\frach$ into itself if and only if  $
 \tau_{b s}(\g)=\g,\  \forall s\in \lat$. On the other hand, \eqref{TsTt'} implies $\tau_{b s} \tau_{b h}(\g)=\tau_{b h}\tau_{b s}(\g)=\tau_{b h}(\g),\  \forall  s\in \lat, h\in \R^2$, so \eqref{frachb-invar} follows. 
\end{proof}
Note that \eqref{usuh'} shows that $u_{b h}$ does not leave $\frach$ invariant. 

 The generators of the magnetic translations, $u_{bs}$, defined in \eqref{ubs'}, and their properties are described in the following
\begin{lemma}[\cite{AHS}]\label{lem:gen-mt}
Let $p_{b} = -i\nabla_{a_b}$ and $\pi_{b} = -\bar p_{ b}:=- \cC p_{ b} \cC$, with the components $p_{b i}$ and $\pi_{b i}$. Then 
\begin{enumerate}
	\item $-i\partial_{s_j} \mid_{s=0} u_{b s} =  \pi_{b j}$; 
	\item $[\pi_{b i}, p_{b j}] = 0$ and therefore $[u_{bs}, p_{b j}] = 0$. 
\end{enumerate}
\end{lemma} 
    
 To begin with, we define the operator $\g$ on $L^2_{\rm loc}(\R^2)$   which could be further specified as $\frach$. 

 For an operator $A$ on $\frach$, we can define the integral kernel $A'(x, y)$ by the relation
 \begin{align}\label{A-A'-rel}\lan g\otimes \bar f, A'\ran_{\frach\otimes \frach} 
=\lan g, A f\ran_{\frach},\ \quad 
\forall f, g\in \frach. \end{align}

In this section, it is convenient to use the notation den$(A)$ for the one-particle density $\rho_A$. 
 For an operator $A$ on $\frach$, s.t. $f A$ and $A f$ are  trace-class, the density den$(A)\equiv \rho_A$ obeys the  relation  
 \begin{align}\label{den-def}\int f \den (A)=\Tr( f A),\ \quad \forall f\in L^\infty.\end{align} 
which can be also used as a definition of den$(A)$. (Recall our convention $\int\equiv \int_{\Om}$ for an arbitrary and fixed lattice cell $\Om$.) With this notation, we have e.g. $j(\g, a) :=\den([-i \na,\g]_+)$.

We assume that all operators below are originally defined on $L^2_{\rm loc}(\R^2)$ (or on a local Sobolev space). This allows us to define compositions and commutators of operators some of which do not leave $\frach$ invariant. 

Recall that we say an operator $A$ to be magnetically translation invariant ({\it mt-invariant}, for short), if and only if it  satisfies $\tau_{b h} A = A,\ \forall h\in \R^2$. Our key result here is the following
\begin{proposition}\label{prop:mt-invar-oprs} 
(i) For a trace-class, mt-invariant operator $A$ on $\frach$, den$(A)$ is constant. (ii) If, in addition, $ \tau^{\rm refl} A =- A$ (with  the reflection operator $\tau^{\rm refl}$ defined before \eqref{even}),   then den$(A)=0$.
\end{proposition}
 
Recall that $\nabla_a:=\nabla - i a$. We derive Proposition \ref{prop:mt-invar-oprs} from the following 
\begin{lemma}\label{lem:grad-den}
For any linear vector field $a$ and any integral operator $A$ on $\frach$,  $[\n_{a}, A]$  leaves $\frach$ invariant (though $\frach$ is not invariant under  $\nabla_a$) and 
 \begin{align}\label{grad-den}\nabla \den(A)  = \den([\nabla_a, A])\end{align}
\end{lemma}
\begin{proof}
  Since $a$ is linear, the invariance of $\frach$ under  $[\n_{a}, A]$  is straightforward. 
We prove \eqref{grad-den}.
We have  \begin{align}\label{den-commut}\den ([\nabla_{a}, A]) =\den ([\nabla_{a_b}, A]) + i\den([a_b-a, A]).\end{align} Since $\nabla_{a_b}$ leaves $\frach$ invariant, we can use  the cyclicity of the trace to compute  
\begin{align}
\notag	\int f \den ([\nabla_{a_b}, A]) &= \Tr_{\frach}(f[\nabla_{a_b}, A])\\ 
\label{den-commut2}		&= -\Tr_{\frach}(\nabla f A)= -\int \nabla f \den(A). 
\end{align}
For any (linear) vector field $c$, the integral kernel of $[c, A]$ is $(c(x)-c(y))A(x, y)$ (see \eqref{A-A'-rel}). Hence, we see that $\den([c, A])=0$. 
 Combining this with \eqref{den-commut} and \eqref{den-commut2} gives $\int f \den ([\nabla_{a}, A])= -\int \nabla f \den(A)$ for any $f\in L^\infty(\R^2)$ and $\lat$-periodic. Hence $\den ([\nabla_{a}, A])=\nabla  \den[A]$. 
\end{proof}

\begin{proof}[Proof of Proposition \ref{prop:mt-invar-oprs}]Since by Lemma \ref{lem:gen-mt}, 
 $\tau_{b h}$ is generated by the map $A\ra i [\pi_{b}, A] $, the mt-invariance of $A$ implies that $[\pi_{b}, A]=0$. (Though $\pi_{b}$ does not leave $A$ invariant, $[\pi_{b}, A]$ does.) This and \eqref{grad-den} yield that $\nabla \den(A) =  \den([\pi_b, A])=0$ and therefore den $ A$ is constant.
\end{proof}

\begin{rem}\label{rem:another-pf-mt-inv-sol} {\em   
 The argument above proving \eqref{grad-den} establishes the intuitive fact that the integral kernel of the operator $[\nabla_a, A]$ acting on $\frach$ is same as the integral kernel of this operator acting on $L^2(\Om)$, which is $(\nabla_{a x}+\overline{\nabla_{a y}})A'(x, y)=(\nabla_{ x}+\nabla_{ y})A'(x, y)$ and consequently $\den([\nabla_a, A])=(\nabla_{ x}+\nabla_{ y})A'(x, y)|_{x=y}=\n \den(A)$.}

\end{rem}
\begin{lemma} \label{lem:gbExistenceAndUniqueness}
\eqref{eqn:NormalBdGPart1}-\eqref{eqn:NormalBdGPart2} have a solution $(\gamma_b, a_b)$, where $\gamma_b \geq 0$ and $\gamma_b$ is an mt-invariant operator on the space $\frach$, 
 if and only if the fixed point problem 
\begin{align}\label{xi-fp}
	\xi = (\int v) \den(f_{T}(h_{ \ab  \mu} + \xi))(0), 
\end{align}
where $h_{ \ab  \mu}:=-\Delta_{ a_b } - \mu$ and  $\xi$ is a real number, has a solution.  (The operator $f_{T}(h_{ a_b } + \xi)$ is well-defined and real since $h_{ a_b }$ is self-adjoint on the space $\frach$.)
\end{lemma}
We show in Appendix \ref{sec:xi-fp} that the fixed point problem \eqref{xi-fp} has a unique solution, 
provided 
the first condition of Theorem \ref{thm:norm-state-exist} holds.  
\begin{proof} Let $\g$ be an even, mt-invariant operator. Then  
$\tilde\g=-i \COVGRAD{a_b}\g$ is an mt-invariant operator and odd.  Applying Proposition \ref{prop:mt-invar-oprs} to $\tilde\g=-i \COVGRAD{a_b}\g$ 
  gives that  $ j(\g, \ab)=0$ and therefore, since $\curl^*\curl \ab=0$, the pair $(\g, \ab)$ satisfies
  \eqref{eqn:NormalBdGPart2}. 
 Hence $(\g, \ab)$ solves  \eqref{eqn:NormalBdGPart1}-\eqref{eqn:NormalBdGPart2} if and only if $\g$ solves	
\begin{align}
	& \gamma = f_{T}( h_{\g \ab \mu}).  \label{eqn:NormalBdGPart1'}
\end{align}

Now, we solve \eqref{eqn:NormalBdGPart1'} for magnetically translation invariant $\g$'s. We treat this equation 
  as a fixed point problem\footnote{If both the direct and exchange self-interactions are dropped from $h_{\g a \mu}$, then the latter equation gives $\g_{T b }$ directly: $\g_{T b }= f_{T}(h_{ \ab  \mu})$, where, recall, $h_{ \ab  \mu}:=-\Delta_{\ab} - \mu$. 
}. This problem simplifies considerably since we dropped the exchange term, as in this case, $h_{\g \ab \mu}$ becomes $h_{ \ab  \mu} + v*\rho$, where, recall, $h_{ a_b  \mu}:=-\Delta_{ \ab } - \mu$. Using this and applying den to \eqref{eqn:NormalBdGPart1} gives the equation for $\rho=\den\g$: 
\[\rho=\den(f_{T}(h_{ \ab  \mu} + v*\rho)).\] 
 Furthermore, by Proposition \ref{prop:mt-invar-oprs}, $\rho=\den\g$ and $\den(f_{T}(h_{ a_b  \mu} + v*\rho)$ are constant functions and therefore $\xi = v*\rho=\int v \rho(0)$ is a real constant satisfying the fixed point equation \eqref{xi-fp}.

 To summarize, we have shown that if an  mt-invariant $\g$ solves \eqref{eqn:NormalBdGPart1'} (i.e. \eqref{eqn:NormalBdGPart1} with $a=\ab$), then $\xi = v*\den\g=\int v \rho(0)$ is a real constant solving \eqref{xi-fp}.

\DETAILS{{\bf(this seems redundant)} $***$ Conversely, if $\gamma$ solves the BdG equation \eqref{eqn:NormalBdGPart1}, then $\xi = v*\den(\gamma)$ satisfies the equation  
\begin{align}
	\xi 
	&= v*\den(f_{\rm FD}((h_{ a_b } + v*\den(\gamma))/T)) \\
		&= v*\den(f_{\rm FD}((h_{ a_b } + \xi)/T)).
\end{align}
So this $\xi$ is solution to \eqref{xi-fp}. 
This proves Lemma \ref{lem:gbExistenceAndUniqueness}.  $***$} 

Now, in the opposite direction, suppose that a real $\xi$ solves \eqref{xi-fp} and define \begin{align} \label{gam-def}\gamma := f_{T}(h_{ a_b  \mu} + \xi).\end{align} 
Since $f_{T} > 0$, $h_{ a_b }$ is self-adjoint and $\xi$ is real, we have that $\gamma \geq 0$. Since $\xi=v*\den(f_{T}(h_{ a_b  \mu} + \xi))=v*\rho_\g$, where, recall,  $\rho_\g\equiv \den(\gamma)$,  \eqref{gam-def} becomes 
\begin{align} \label{eqn:gammaDefInTermsOfxi}
	\gamma 
		&= f_{T}(h_{ a_b  \mu} + v*\rho_\g).
\end{align}
Hence $\gamma$ satisfies \eqref{eqn:NormalBdGPart1'}, i.e. \eqref{eqn:NormalBdGPart1} with $a=\ab$. 
\end{proof}

Since, as we show in Appendix \ref{sec:xi-fp}, the fixed point problem \eqref{xi-fp} has a unique solution, 
 provided the potential $v$ satisfies 
  the first  condition of Theorem \ref{thm:norm-state-exist}, 
we obtain an unique magnetic translation invariant solution of  \eqref{eqn:NormalBdGPart1}-\eqref{eqn:NormalBdGPart2}, under the same condition. 

\DETAILS{Assume that $\gamma_1,\gamma_2$ are two solutions to \eqref{eqn:NormalBdGPart1}. Then we may form the corresponding $\xi_i = v*\den(\gamma_i)$ for $i=1,2$. Uniqueness of solution of equation \eqref{xi-fp} dictates that $\xi_1=\xi_2$. Therefore,
\begin{align}
	\gamma_1 =& f_{\rm FD}((h_{ a_b }+v*\den(\gamma_1))/T) \\
		=& f_{\rm FD}((h_{ a_b }+\xi_1)/T) \\
		=& f_{\rm FD}((h_{ a_b }+\xi_2)/T) \\
		=& f_{\rm FD}((h_{ a_b }+v*\den(\gamma_2))/T) = \gamma_2.
\end{align}}

The uniqueness part of Lemma \ref{lem:gbExistenceAndUniqueness} is strengthened and extended to the second condition in the following 
\begin{lemma} \label{lem:gbExistenceAndUniqueness'}
A solution of \eqref{eqn:NormalBdGPart1}-\eqref{eqn:NormalBdGPart2} is unique among pairs $(\gamma, a)$ with $\gamma$ mt-invariant.  
\end{lemma}
\begin{proof} First observe that, by Proposition \ref{prop:mt-invar-oprs}, $\gamma(x,x)$ is constant and the term $\Re(-i\nabla_{a_b}\gamma)(x,x)$ vanishes. We decompose $a=a_b+a'$, where $a'$ is defined by this expression. Using \eqref{eqn:NormalBdGPart2} and $\curl^* \curl a_b = 0$, we see that
\begin{align}\label{a'-eq}
	\curl^* \curl a' = -\gamma(0,0)a'.
\end{align}
Since $h_{ a_b }+\xi$ is bounded below and $f_{T}$ is strictly positive and monotonically increasing, we see that $\gamma = f_{T}(h_{ a_b }+\xi) \geq c > 0$. Thus $\gamma(0,0) > 0$. 
 Multiplying both sides of \eqref{a'-eq} by $a'$ and integrating, we find
\begin{align}
	\int |\curl a'|^2 + \gamma(0,0) \int |a'|^2 = 0.
\end{align}
Since  $\gamma(0,0) > 0$, this implies that $a' = 0$. 
Hence $a=a_b$ and therefore \eqref{eqn:NormalBdGPart1} shows that $\gamma$ is a function of $-\Delta_{a_b}$ and therefore mt-invariant. Hence, we can conclude uniqueness by Lemma \ref{lem:gbExistenceAndUniqueness} 
 and the proof is complete.
 \end{proof}
Lemmas \ref{lem:gbExistenceAndUniqueness}, \ref{lem:gbExistenceAndUniqueness'} 
and \ref{lem:xi-fp} of Appendix \ref{sec:xi-fp} imply Theorem \ref{thm:norm-state-exist},  under  the first  
 condition.

To prove  Theorem \ref{thm:norm-state-exist} under the second condition, i.e. for $\|v\|_\infty<\infty$, 
 we use the variational approach. Let $ 
		 D^s_\nu 	=\big\{\g\in  I^{s, 1},\  0\le \gamma=\g^* \le 1,\ 
		 \g \text{ 
		 satisfies \eqref{even}},\\  
	\tr\g=\nu 
	\big\} $ (see \eqref{gam-al-prop} and \eqref{Is1}). Recall, that 
 $S(\g)$ is the entropy defined in  \eqref{S-expr}-\eqref{g-s-def} and define the free energy functional
\begin{align} \label{F1}
	F_1 (\g) = \Tr\big(h_{a_b} \gamma \big)	 + \frac12 \int \rho_\g(v* \rho_\g)  
 - TS(\g),
\end{align}
 on $D^1_\nu$, setting $F (\g) =- \infty$,  if $S(\G) = \infty$. (To compare with \eqref{FT-def}-\eqref{energy}, $\int \rho_\g(v* \rho_\g)=\Tr\big((v* \rho_\g) \gamma \big)$.) 
 We show below the following 
\begin{proposition}\label{prop:F1-min}
 Assume $v$ is $\lat$-periodic and satisfies $\|v\|_\infty<\infty$. Then  the functional $F_1 (\g)$ on $D^1_\nu$  has a minimizer; minimizers of  $F_1 (\g)$ satisfy $0< \gamma < 1$. 
\end{proposition}
   It follows then, by a special case of Theorem \ref{thm:BdG=EL}, that minimizers of $F_1 (\g)$ satisfy the Euler-Lagrange equation, which is equivalent to \eqref{eqn:NormalBdGPart1'}. Using the latter equation, the solution could be bootstrapped from $ I^{1, 1}$ to $ I^{2, 1}$, giving Theorem \ref{thm:norm-state-exist}.  
   $\Box$

\begin{proof}[Proof of Proposition \ref{prop:F1-min}] First, we show that $F_1 (\g)$ is coercive, namely, that \begin{align} \label{F1-lbnd} F_1 (\g) \ge  \frac12 \|\gamma\|_{I^{1, 1}}- C_{\nu, T}.  
\end{align} 
Indeed, by the definition,  $\|\gamma\|_{I^{1, 1}}=\Tr(h_{a_b} \gamma)$ and the elementary estimate 
\begin{align} \label{DI-est}
	| \int \rho (v* \rho)| \le \|v\|_\infty \|\rho\|_1^2,\end{align}
	 we have 
 \begin{align} \label{F1-deco} 
F_1 (\g) \ge  \frac12 \|\gamma\|_{I^{1, 1}} + f(\g)- \|v\|_\infty \|\rho_\g\|_1^2,\end{align}
where   
 $f(\g):= \frac12 \|\gamma\|_{I^{1, 1}}  - TS(\g)$.
We minimize the functional $f(\g)$ on the r.h.s. on the set $ I^{1, 1}_\nu:=\{\g\in I^{1, 1}: \tr \gamma=\nu\}$. Since $f(\g)$ is convex and the constraint $\Tr\g=\nu$ is linear, each solution to the standard Euler-Lagrange equation $d f(\g) - \mu d \Tr\g =0$ (written in terms of the G\^{a}teaux derivatives), where $\mu$ is the Lagrange multiplier,  is a global minimizer. The latter equation is computed to be $ \frac12  h_{a_b} - T\ln\big(\frac{\gamma}{\one -\g}\big)-\mu \one=0$.
 Solving this equation 
  gives the minimizer 
 \begin{align}\label{gam-muT}
	\gamma_{\mu, T} = f_{T}(\frac12 h_{a_b} -\mu),
\end{align}
for  $\mu$ such that $\Tr \gamma_{\mu, T} = \nu$.  By the inverse 
  function theorem, the latter equation has a solution, $\mu=\mu(T, \nu)$, for $\mu$. This shows $ - C_{\nu, T}:=\inf\{f(\g):\ \g \in I^{1, 1}_\nu\}=f(\g_{\mu, T})>-\infty$, with $\mu=\mu(T, \nu)$,  which, together with \eqref{F1-deco}, implies  \eqref{F1-lbnd}. 

Next, we show that $F_1(\g)$ is weak lower semi-continuous. To this end, we pass from the positive, trace class operators $\g$ to the Hilbert-Schmidt ones, $\ka:=\sqrt \g$. Note that $\g\in I^{1,1}, \g\ge 0 \Longleftrightarrow \ka:=\sqrt \g \in I^{1,2}$. Thus we consider $F_1(\g)=: F'(\ka)$ 
on the space $\tilde D^1_\nu 
:=\big\{\ka\in  I^{1, 2},\  0\le \ka=\ka^* \le 1,\ 		 \ka \text{ satisfies \eqref{even}},\ \tr\ka^2=\nu\big\}$. 
 The first term on the r.h.s. of \eqref{F1} satisfies $ \Tr(h_{a_b}  \gamma)=\| \ka \|_{I^{1, 2}}^2$ 
 and is quadratic in $\ka$. Hence, it is $\|\cdot \|_{I^{1,2}}$-weakly lower semi-continuous.
For the second term, we use the inequalities \eqref{DI-est} and 
\begin{align} \label{rho-est}
	 \|\rho_{\g'}-\rho_{\g}\|_1\le \| \ka' - \ka\|_{I^{2}}(\| \ka' \|_{I^{ 2}}+\| \ka \|_{I^{2}})\end{align} to show that it is also lower semi-continuous.  The third term on the r.h.s. of \eqref{F1},  $-T S(\g)$, is lower semi-continuous, 
  by Lemma \ref{lem:Slsc} of Appendix \ref{sec:entropy}. Hence $ F_1(\g)=: F'(\ka)$ is lower semi-continuous.
  
  Finally, we observe that the set $\tilde D^1_\nu$ is  closed in $I^{1,2}$ under the weak convergence. 

 With the results above, the proof of existence of a minimizer is standard. To avoid repetitions, we refer to the second paragraph after Lemma \ref{lem:Flsc}, where this is done in somewhat more complicated notation. 

Now, we establish properties of minimizers  $\g_*=\ka_*^2 $. By \eqref{F1-lbnd}, we have $S(\g_*)=\Tr g(\g_*)<\infty$. 
This and the fact that $g(\g)\ge 0$ imply that $g(\g_*)$ is trace class.
%

Finally, a simplified version of the proof of Lemma \ref{lem:eta*EVs} shows that 
$0$ and $1$ are not eigenvalues of $\g_*$ and therefore, $0 < \g_* < 1$.\end{proof}


\section{Stability/instability of the  normal states for small $T$ and $b$: Proof of Propositions  \ref{prop:T-b-large} and \ref{thm:T-b-small}}  
 \label{sec:normal-stab/instab}

 \begin{proof}[Proof of Proposition \ref{prop:T-b-large}]
 Recall that $K_{Tb}=T f(h_x/T, h_y/T)$, where $f(u, v):=\frac{u+v}{\tanh(u)+\tanh(v)}$ and $h_z$ is the operator $h_{Tb}$, defined in Proposition \ref{prop:FT''}, acting on the variable $z$. By Lemma \ref{lem:f-low-bnd} below $f(u, v)\ge 1$. (A  weaker bound $f(u, v)\ge \frac14$ which suffices for us could be easily proved directly.) \DETAILS{Indeed, assume for simplicity that $x, y\ge 0$ and write 
\begin{align}\label{f-fn}
	f(x, y)= \frac{(x+ y)(1+e^{-2x})(1+e^{-2y})}{2(1-e^{-2(x+y)})}.
\end{align} 
Since $1+e^{-u}\ge 1$ and $1-e^{-u}\le u$, for $u\ge 0$, this gives $f(u, v)\ge \frac14$ and therefore} Hence 
\begin{align}\label{KTb-bnd}
	K_{Tb}\ge  T,
\end{align}  
which implies that 
 $L_{Tb}\ge  T-\|v\|_\infty$  and consequently, Proposition \ref{prop:T-b-large} follows.  
 \end{proof} 
 
  
  \begin{proof}[Proof of Proposition \ref{thm:T-b-small}]

 We use the Birman-Schwinger principle to show that $L_{Tb}$ has a negative eigenvalue. 
   Set $w^2 = -v \ge 0$ so 
  that $L_{Tb} = K_{Tb} - (w^\sharp)^2$, where, recall,  $(w^\sharp\, \al) \, (x;y):= w(x -y)\al (x;y)$. 
  
  By the Birman-Schwinger principle,  $L_{Tb}$ has a negative eigenvalue $-E$ if and only if $G_{Tb}(E):= w^\sharp(K_{Tb}+E)^{-1} w^\sharp$  has the eigenvalue $1$ for some $E > 0$ (see e.g. \cite{GS}). 
By \eqref{KTb-bnd}, we have  $G_{Tb}(E)\ge 0$  for all $E\ge 0$.  Moreover, since $(K_{Tb}+E)^{-1}$ is continuous and monotonically decreasing in $E\ge 0$  and vanishing as $E\ra \infty$, so is $G_{Tb}(E)$. Hence, it suffices to show that $G_{Tb}:=G_{Tb}(0)$ satisfies the estimate $\|G_{Tb}\|>1$, which we now prove.

 
Recall that $K_{Tb}=T f(h_x/T, h_y/T)$, where $f(s, t):=\frac{s+t}{\tanh(s)+\tanh(t)}$ and $h_z$ is the operator $h_{Tb}$, defined in Proposition \ref{prop:FT''},  
 acting on the variable $z$. 
Since $h_{Tb}$ satisfies $h_{Tb}\ge - \mu'$, for some $\mu'>\mu$, it suffices to consider $f(s, t)$ for $s, t\ge -\mu'$. A simple estimate 
\begin{align}\label{f-est} 
f(s, t)\ls 1+|s+t|, \end{align} 
for $s, t\ge -\mu'$, which follows from Lemma \ref{lem:f-low-bnd} below, yields $K_{Tb}\ls T +|h_x + h_y|$. This implies the inequality
\[G_{Tb}\ge 
w^\sharp (T+ |h_x + h_y|)^{-1} w^\sharp \ge 0.\]

Since we omit  the 
exchange term $-v^\sharp \gamma$,  the operator $h_{Tb}$ is of the form $h_{ a_b  \mu}:=-\Delta_{a_b}+ v*  \rho_{\gamma_{Tb}} -\mu$. 
By Proposition \ref{prop:mt-invar-oprs}, $\rho_{\gamma_{Tb}}$ is a constant function and therefore $\xi:= v*\rho=\int v \rho(0)$ is a real constant. Hence  $h_{ a_b  \mu}:=-\Delta_{a_b}+ \xi -\mu$. 
 
Since the gaps between  the eigenvalues $\lambda_n=b(2n+1)$  of  $-\Delta_{a_b}$ on $\frach$ are equal to $b$, we can choose $m$ s.t. $|\lambda_m+ \xi-\mu|\ls b$.

 Recall that $L_{Tb}$ acts on the space of the Hilbert-Schmidt operators which can be identified through their integral kernels with $\mathfrak{h} \otimes \mathfrak{h}$. 
 Let $\phi_m$ be the normalized eigenfunctions of $-\Delta_{a_b}$ corresponding to  the eigenvalues $\lambda_m=b(2m+1)$. We take 
$u:=c(w^\sharp)^{-1}(\phi_m\otimes \phi_m)$, where $c=\|(w^\sharp)^{-1}(\phi_m\otimes \phi_m)\|^{-1}$, so that $\|u\|=1$.  By analyzing $f(s, t)$, $s, t\ge -\mu'$,  separately in several domains, we obtain
\begin{align}\label{f-est} \lan u, G_{Tb} u\ran&\gs (T+|\lambda_m+ \xi-\mu|)^{-1}\|(w^\sharp)^{-1}(\phi_m\otimes \phi_m)\|^{-2}\notag\\
&\gs (T+b)^{-1}\|(w^\sharp)^{-1}(\phi_m\otimes \phi_m)\|^{-2}. \end{align}
(\eqref{f-est} also follows from the stronger, but more involved, Lemma \ref{lem:f-low-bnd} below.)
Now, write $\phi_m (x)= \sqrt{b}\phi_m^0(\sqrt b x)$, where $\phi_m^0 (x)$ is independent of $b$. Furthermore, by  the assumption on $v=-w^2$, we have $w(x-y)\gs (1+|x-y|)^{-\kappa/2}, \kappa<2$. 
The last two relations, together with  the inequality $(a+b)^\kappa\le 2^\kappa (a^\kappa+b^\kappa)$,  imply \[\|(w^\sharp)^{-1}(\phi_m\otimes \phi_m)\|^2 \ls \int(1+|x-y|^{\kappa})|\sqrt{b}\phi_m^0(\sqrt b x) \sqrt{b}\phi_m^0(\sqrt b y))|^2dxdy.\]
Changing the variables of integration as $x'=\sqrt b x, y'=\sqrt b y$, we find  $\|(w^\sharp)^{-1}(\phi_m\otimes \phi_m)\| \ls  b^{-\kappa/4}$, which in turn gives  $\lan u, G_{Tb} u\ran \gs  (T+b)^{-1} b^{\kappa/2}\ra \infty$ as $ b\ra 0,$ provided $ T\ls b^\sigma, \sigma >\kappa/2$.

Thus we have shown that $\|G_{Tb}\|$, or the largest eigenvalue of $G_{Tb}$, can be made arbitrarily large if $T$ and $b$ are sufficiently small, which, by the Birman-Schwinger principle, proves Proposition \ref{thm:T-b-small}. 
 %
\end{proof}
Bounds on the function $f(s, t):= \frac{s+ t}{\tanh(s)+\tanh( t)}$ used in the proof above could be proved directly; they also follow from the following

\begin{lemma}\label{lem:f-low-bnd}
The function $f(s, t):= \frac{s+ t}{\tanh(s)+\tanh( t)}$ has the minimum $1$ achieved at $s = t=0$.
\end{lemma}
\begin{proof} 
To find minimum of $f$, we look for its critical points. We let $g(s, t) = \tanh(s) + \tanh(t)$ and compute
\begin{align}
	\nabla f = \frac{1}{g}( 1- f(s, t) \sech^2(s), 1-f(s, t) \sech^2(t)).
\end{align}
Setting $\nabla f = 0$ yields that
\begin{align}\label{f-inv}
	f(s, t)^{-1}=  \sech^2(s) \ \text{ and }\ 
	f(s, t)^{-1}= \sech^2(t) .
\end{align}
It follows that $\sech^2(s) = \sech^2(t)$ and therefore either $s=t$ or $s=-t$. If $s = -t$, then $f(s, t) = \sech^{-2}(s)$, which has a single critical point - minimum - 
 $s = 0$. Hence in this case we have  a single critical point - minimum - at $s = -t=0$ and $f(0,0) = 1$. If $s=t$, then \eqref{f-inv} becomes
\begin{align}
	\tanh(s) = s \, \sech^2(s), \text{ or equivalently, }\
	\sinh(u)\cosh(s) = s .
\end{align}
This is equivalent to $\sin(2s) = 2s$ which implies that $s=0$. Hence the minimum is reached at $s = 0$ and $f(0,0) = 1$.
 \end{proof}

 
\section{The existence of the vortex lattices} \label{sec:vort-latt-exist} 
In this section, we prove Theorem \ref{thm:BdGExistence} on existence of the vortex lattice solutions to the BdG equations with  arbitrary but fixed lattice $\lat$ and  first Chern (vortex) number $c_1(\chi^\lat) = n\in \Z$, i.e. the integer $n$ entering \eqref{b-quant}  (see \eqref{mf-quant} and \eqref{c1}). 
Recall that  $b=\frac{2\pi n}{|\lat|},\ n \in  \Z$, see \eqref{b-quant}.

Recall that 
we drop the exchange self-interaction term 
  $v^\sharp \gamma$. We minimize the resulting energy for $\Tr(\gamma)$ fixed. Hence we omit the term $- \mu \Tr(\gamma)$ in \eqref{FT-def}. Also, in this section, we display the domain of integration $\Om$ which is an arbitrary but fixed fundamental cell of the lattice $\lat$. Thus, with the  notation $h_{a}:= -\Delta_{a}$, the free energy functional $F_{T}(\G, a)$ in \eqref{FT-def} becomes 
\begin{align} \label{cF}
	\F (\G, a) = 
\Tr\big(h_{a} \gamma \big)	& + \frac12 \int \rho_\g(v* \rho_\g)\notag\\ & +\frac{1}{2} \Tr\big( \s^* v^\sharp \s \big) 
+ \int_{\Omega} dx  |\curl a |^2 
 - TS(\G),
\end{align}
where, recall, 
the entropy $S(\G)$ is defined in  \eqref{S-expr}-\eqref{g-s-def}. The functional $\F$ is defined on $\mathcal{D}^1_\nu \times \frachvec^1$, 
 if $S(\G) < \infty$. Otherwise, we set $\F (\G,a) = \infty$.  
  



 Theorem \ref{thm:BdGExistence}(i) follows from  Theorem \ref{thm:BdG=EL}  and the following  
  \begin{theorem} \label{thm:ExistMin}
 Fix a lattice $\lat$ and a Chern number $c_1(\chi^\lat)=n$. Assume that $T > 0$ and $\|v\|_\infty < \infty$. 
  There exists a finite energy minimizer $(\G_*, a_*) \in \mathcal{D}^1_\nu \times \frachvec^1$ of the functional $\mathcal{F}(\G,a)$ on the set $\mathcal{D}^1_\nu \times \frachvec^1$. This minimizer has the equivariance and the flux quantization properties, \eqref{VL:equiv} and \eqref{mf-quant}, satisfies $0 < \G_*< 1$ and is s.t. $g(\G_*)$ (see \eqref{S-expr}) is trace class. Furthermore, the minimizer $(\G_*,a_*)$ can be chosen to be even, i.e. satisfying \eqref{even}. \end{theorem}


\begin{proof}[Proof of Theorem \ref{thm:ExistMin}] 
\DETAILS{
 We pass from the vector potential $a$ to $e:=a-a_b$.   Furthermore,
 using $a = a_b+e$, $\div a_b = b$ and $\int_{\Omega} =0$, we compute 
\begin{align}
	 \int_{\Omega} |\curl a|^2 =&  
	 \int_{\Omega} |\curl e|^2 + b^2 |\Om|. 
\end{align}
(Recall that $\Om$ is an arbitrary fundamental cell of the lattice $\lat$.)  
We also assume  $\tr\g=\nu$. Consequently, 
 consider, instead of \eqref{cF}, the equivalent functional 
   \begin{align} \label{F-def}
	F (\g, \al, e)& = \mathcal{F} (\G, a_b+e) - b^2 |\Omd|\\
 \label{F} &=\Tr\big( \ka h_{a_b+e} \ka \big)	 +\frac{1}{2} \Tr\big( \s^* v^\sharp \s \big) 
+ \int_{\Omega}  |\curl e|^2  
 - TS(\G)
\end{align}
on the space $I^{1,1}\times I^{1,2}\times \frachvec^{1}$ with  the norm $\|\g \|_{I^{1,1}} + \|\al \|_{ I^{1,2}}+\|e \|_{\frachvec^{1}}$ and with the side conditions $0\le \G \le 1$ and $\tr\g=\nu$.}

We will use standard minimization techniques 
proving that $\mathcal{F} (\G, a)$ 
 is coercive and weakly lower semi-continuous, and $\mathcal{D}^1_\nu \times \frachvec^1$ weakly closed. 

\textbf{Part 1: coercivity.}  
The main result of this step is the following proposition:
\begin{proposition} \label{pro:LowerBound}
Let $T > 0$, $e:=a-a_b$,  $\g\in I^{1, 1}$ and $\al\in I^{1, 2}$, with $\tr \g=\nu$ and $0\le \eta \le \one$. 
Then 
\begin{align}\label{cF-lower-bnd'}
\mathcal{F} (\G, a_b+e)	& \geq 
c' [\|\g\|_{I^{1,1}}/ \nu]^r+ c\|e\|_{H^1}^2 - C\\
\label{cF-lower-bnd}& \geq 
\frac14 c'  [(\|\gamma \|_{I^{1,1}} + \|\al \|_{ I^{1,2}}^2)/ \nu]^r + c\|e\|_{H^1}^2 - C,\end{align}  
 for any $0<r<1$ and for suitable constants $c', c, C > 0$, with $c', c > 0$ independent of $\nu$, $T, v$ and $\lat$ and $C$ depending on $\nu$, $T, \|v\|_{\infty}$ and $|\lat|$. 
\end{proposition}
\begin{proof}[Proof of Proposition \ref{pro:LowerBound}]  We begin with estimating the entropy term $- T S(\G)$ (cf. \cite{HaiSei}). Recall definition  \eqref{S-expr}-\eqref{g-s-def} of $S(\G)$, which we reproduce here 
\begin{align}
\label{S-expr'}	& S(\G) := \Tr(s(\G))= \Tr(g(\G)),\\
 \label{g-s-def'}	&  g(\G): = -\G\ln\G - (1-\G)\ln(1-\G),\ s(\G ):= -2\G\ln\G.
\end{align} 
Note that since $s(\G )\ge 0$ for $0\le \G\le \one$, we have that $S(\G)\ge 0$. We  define  the relative entropy
\begin{align}\label{RelatEntropy}	
S(A|B) = \Tr(s(A|B)), \quad	& s(A|B) := A(\ln A - \ln B). 
\end{align}
We define the diagonal operator-matrix $\G_0$  and recall the off-diagonal one $\phi$:
\begin{align}
\label{Gam0Phi} 
\G_0 := \left( \begin{array}{cc} 
	\g & 0 \\  
	0 & \one-\bar\g 
\end{array} \right), \ 	\phi(\beta) :=\left( \begin{array}{cc}
											0 & \beta \\
											\beta^* & 0
										\end{array} \right). 
\end{align}

\begin{lemma}[cf. \cite{FHSS}] \label{lem:FHSSEntropyBnd} We have for $\G = \G_0 + \phi(\alpha)$,
\begin{align}\label{S-relatS}
	 S(\G) =  S(\Gam_0) - S(\G | \Gam_0)\le   S(\Gam_0). \end{align} 
\end{lemma}
\begin{proof}  We note that for $\G := \G_0 + \phi(\alpha)$,
\begin{align}
	 \G \ln \G & 
	 - \G_0 \ln \G_0  
	= \G \ln \G - \G \ln \G_0 + \G \ln \G_0  
		- \G_0 \ln \G_0\\ 
	&= s(\G,\G_0) + (\G -\G_0)\ln \G_0\\ 
\label{elne}	&= s(\G,\G_0) + \phi(\alpha) \ln \G_0. 
\end{align}
Since the last term, $\phi(\alpha) \ln \G_0$, 
 in \eqref{elne} is off-diagonal and therefore has zero trace, the first equation in \eqref{S-relatS} follows.

The inequality in \eqref{S-relatS} follows from  Klein's inequality and the fact that $\Tr \G =\Tr \G_0$.\end{proof}
With the definitions  \eqref{S-expr'}-\eqref{g-s-def'} and \eqref{Gam0Phi}, we have that \eqref{S-relatS} implies
\begin{align} \label{SGam0} 
S(\G)\le S(\G_0)  = \Tr(g(\g)),\ g(x) = -[x\ln x +(1-x)\ln (1-x)].
 \end{align}

Next, we estimate the second ($\al$-) term on the r.h.s. of \eqref{cF}. To this end, we 
 bound $\alpha$ by 
  $\gamma$ via the constraint $0 \leq \G \leq 1$ using the following result (see \cite{BBCFS} and references therein):

\begin{lemma}\label{lem:al-gam-bnd}
The constraint $0 \leq \G \leq 1$ implies that
\begin{enumerate}
\item 
$\alpha^*\alpha \leq \bar{\gamma}(1-\bar{\gamma})$ and $\alpha \alpha^* \leq \gamma(1-\gamma)$.
\item $\Tr (M\alpha \alpha^* M^*) \leq \|M\gamma M \|_{1}$ for any operator $M$.\end{enumerate}
\end{lemma}
\begin{proof}
Since $0 \leq \G \leq 1$, then $0 \leq 1-\G \leq 1$ as well and therefore  
 \[0 \leq \G(1-\G) \leq 1.\] 
From \eqref{Gam}, we see that the 1,1-entry of $\G(1-\G)$ is $0 \leq \gamma(1-\gamma) - \alpha \alpha^*$. By considering the 2,2-entry, we have that $\alpha^* \alpha \leq \bar{\gamma}(1-\bar{\gamma})$. 
 Finally, since $1-\gamma \leq 1$, we see that $M\alpha \alpha^* M^* \leq M \gamma M^*$, which 
 completes the proof.
\end{proof}

Since $v$ is bounded, Lemma \ref{lem:al-gam-bnd}(1) gives
\begin{align}\label{al-gam-bnd}	|\Tr(\alpha^* v^\sharp \alpha)| \leq \|v\|_{\infty} \Tr(\alpha^* \alpha)  \leq \|v\|_{\infty} \Tr(\gamma).
\end{align}

Now, using estimates \eqref{SGam0} and \eqref{al-gam-bnd} in expression \eqref{cF} and introducing the notation  $S(\g):=\Tr(g(\g))$, we find
\begin{align} \label{F-bnd1}
\mathcal{F} (\G, a)	 \ge\Tr\big(  h_{a} \g \big)& + \frac12 \int \rho_\g(v* \rho_\g)\notag\\ & + \int_{\Omega}  |\curl a|^2   - TS(\g) - \|v\|_{\infty} \Tr(\gamma).\end{align}
 The r.h.s.  depends only on $\g$ and $a=:\ab+e$ with the only constraint $\Tr \g=\nu$ and we estimate it next.

Next,  we pass from the vector potential $a$ to $e:=a-a_b$.   
 Using $a = a_b+e\in \frachvec^{1}$, $\curl a_b = b$ and $\int_{\Omega}\curl e=0$ (see \eqref{fh-vec}-\eqref{fhb-vec}), we compute 
\begin{align} \label{aH1-eH1}	 \int_{\Omega} |\curl a|^2 =&  
	 \int_{\Omega} |\curl e|^2 + b^2 |\Om|. 
\end{align}
 Since $\int_{\Omega} e = 0$ and $\div e = 0$, the Poincar\'e's inequality shows that, for some $c>0$, 
\begin{align}\label{eH1-est} 2c\|e\|_{H^1}^2\le \int_{\Omega} |\curl e|^2.	
\end{align}

Now, we estimate $\Tr(h_a \gamma)$, where, recall, $a = a_b+e$. 
 Using $\div e = 0$, we write 
\begin{align}
\notag	\Tr(h_{a} \gamma) &=\Tr(h_{a_b}\gamma) + 2i\Tr(e \cdot \nabla_{a_b} \gamma) + \Tr(|e|^2 \gamma) \\
\label{Tr1-est}		\geq & (1-\epsilon) \Tr(h_{a_b}\gamma) - (\epsilon^{-1}-1)\Tr(|e|^2 \gamma), 
\end{align}
for any $\epsilon > 0$.  
For the last term, we claim, for any $r\in (0,1)$ and some constant $C$, the estimate 
\begin{align}\label{e2gam-est}
	0\le \Tr(|e|^2 \gamma)\le C \|\gamma\|_{I^{1, 1}}^{1-r}  (\tr \gamma)^{r}  \|e\|_{H^1}^2,  
\end{align}
where we used that $\Tr(h_{a_b}\gamma)=\|\gamma\|_{I^{1, 1}}$. We prove this estimate for $r=1/2$, which suffices for us.  For general $r\in (0,1)$, see Lemma \ref{lem:e2gam-est} of Appendix \ref{sec:den-est}. Recall the definition $M_b:=\sqrt{h_{a_b}}$.   We use relative bound \eqref{Sob-ineq3} of Appendix \ref{sec:den-est} to find 
\begin{align}\label{e2gam-est'}
	0\le \Tr(|e|^2 \gamma)\ls \||e|^2 M_b^{-s}\|  \|M_b^{s}\ka \|_{I^{2}}  \|\ka\|_{I^{2}} \ls \|e\|_{H^v}^2 \|\ka\|_{I^{s, 2}} \|\ka\|_{I^{0, 2}}\end{align}
with $s>2(1-v)$. The last estimate gives \eqref{e2gam-est} with $r=1/2$.
 
Let $\xi(\g) :=c/(C(\epsilon^{-1}-1)\|\gamma\|_{I^{1, 1}}^{1-r}  \nu^{r})$, where, recall, $\nu=\tr \gamma$ and $r\in (0,1)$, for the constant $C$ coming from \eqref{e2gam-est}, set $\e =\frac14$. 
Then,  \eqref{Tr1-est} and \eqref{e2gam-est}, the inequality $\Tr(h_{a} \gamma)  \geq \del \Tr(h_{a} \gamma),$ for $\del\le 1$, and the relation $\|\gamma\|_{I^{1, 1}}=\Tr(h_{a_b}\gamma)$ 
imply
\begin{align}\label{Trhgam-ineq}
	\Tr(h_{a} \gamma)  \geq & \del[\frac34 \|\gamma\|_{I^{1, 1}} 
		 - c \xi(\g)^{-1}  \|e\|_{H^1}^2]. 
\end{align}
 Now,  taking $\del=\xi(\g)$ and using that $\xi(\g) \Tr(h_{a_b} \gamma)=c/(3C) (\|\gamma\|_{I^{1, 1}}/\nu)^{r}$ gives 
\[\Tr(h_{a} \gamma) \ge 2c'  [\|\gamma\|_{I^{1, 1}}/ \nu]^r  - c  \|e\|_{H^1}^2,\] 
where  $c':= c/(8C)$.  The latter equation implies, in turn, 
\begin{align} \label{E-bnd1} 
&\Tr(h_{a} \gamma) - TS(\gamma)	\ge c' E'(\g) 
+ c'  [\|\gamma\|_{I^{1, 1}}/ \nu]^r - c  \|e\|_{H^1}^2,\\ 
\label{E'-expr}&E'(\g) := [\|\gamma\|_{I^{1, 1}}/ \nu]^r - TS(\gamma).\end{align}
To estimate $E'(\g)$ from below, 
 let $e_k$ be an orthonormal eigenbasis of $h_{a_b}$ and $\lam_k$ be the corresponding eigenvalues. Using that $\|\gamma\|_{I^{1, 1}}=\Tr(h_{a_b} \gamma) = \sum_{k} \lam_k \lan e_k,\gamma e_k \ran$,  regarding $\lan e_k, \gamma e_k \ran /\Tr \gamma$ as a probability measure and applying Jensen's inequality (for $0 < r < 1$ so that $x^r$ is concave), we find
\begin{align}
	(\Tr(h_{a_b} \gamma)/\Tr\gamma)^r	=& \big( \sum_{k} \lam_k \lan e_k,\gamma e_k \ran/\Tr\gamma \big)^r\\ 
		&\geq  \sum_{k} \lam_k^r \lan e_k,\gamma e_k \ran/\Tr\gamma. 
		\end{align}
Since $ \sum_{k} \lam_k^r \lan e_k,\gamma e_k \ran	= \Tr (h_{a_b}^{r} \gamma)$, we have
\begin{align}
	(\Tr(h_{a_b} \gamma)/\Tr\gamma)^r	&\geq  \Tr (h_{a_b}^{r} \gamma)/\Tr\gamma.\end{align}
  Let $ I^{1, 1}_\nu:=\{\g\in I^{1, 1}: \tr \gamma=\nu\}$.
The above estimate, together with \eqref{E'-expr}, restricted to the se $I^{1, 1}_\nu$, gives 
\begin{align}\label{E'-low-bnd}
	E'(\gamma) \geq &  \Tr (h_{a_b}^{r} \gamma)/\nu - TS(\gamma) =:E_\nu(\gamma).\end{align}
We minimize the functional $E_\nu(\g)$ on the r.h.s. on $I^{1, 1}_\nu$. Since $E_\nu(\g)$ is convex and the constraint $\Tr\g=\nu$ is linear, each solution to the standard Euler-Lagrange equation $d E_\nu(\g) - \mu' d \Tr\g =0$ (written in terms of the G\^{a}teaux derivatives), where $\mu'$ is the Lagrange multiplier,  is a global minimizer. The latter equation is computed to be $  h_{a_b}^{r}/\nu - T\ln\big(\frac{\gamma}{\one -\g}\big)-\mu' \one=0$.
 Solving this equation and setting $\mu'= \mu/\nu$ gives the minimizer 
 \begin{align}
	\gamma_{\mu, T\nu} = f_{T}( (h_{a_b}^{r} -\mu)/\nu )
\end{align}
for  $\mu$ such that $\Tr \gamma_{\mu, T\nu} = \nu$.  By the implicit function theorem, the latter equation has a solution, $\mu=\mu(T, \nu)$, for $\mu$. This shows $ - C_{\nu, T}/c':=\inf\{E_\nu(\g): \g \in I^{1, 1}_\nu\}=E_\nu(\g_{\mu, T\nu})>-\infty$, with $\mu=\mu(T, \nu)$,  which, together with \eqref{E'-low-bnd}, implies  
\begin{align} \label{E''-lwbnd}\inf_{\g\in I^{1, 1}_\nu}E'(\g) \geq \inf_{\g\in I^{1, 1}_\nu}E_\nu(\g) =- C_{\nu, T}/c'>-\infty,\end{align} 
 which, together with \eqref{E-bnd1}, implies 
\begin{align} \label{E-bnd2} \Tr(h_{a} \gamma) - TS(\gamma)	\ge & c' \ [\|\gamma\|_{I^{1, 1}}/ \nu]^r- C_{\nu, T} - c  \|e\|_{H^1}^2,\end{align}
with, recall,  $c':=c/(8C)$. Estimates \eqref{F-bnd1}, \eqref{aH1-eH1}, \eqref{eH1-est}, \eqref{DI-est} and \eqref{E-bnd2}, restricted to $I^{1, 1}_\nu, \nu \ge 1,$ give 
 bound \eqref{cF-lower-bnd'}. 
This bound and  Lemma \ref{lem:al-gam-bnd}(2) imply bound  \eqref{cF-lower-bnd}. This completes the proof of Proposition \ref{pro:LowerBound}.
\end{proof}

 We continue with the proof of Theorem \ref{thm:ExistMin}. Eq. \eqref{E-bnd2} shows that it suffices to consider $\g$ with $S(\g)<\infty$ and therefore, by \eqref{SGam0},  $\G$ with 
 \begin{align} \label{S-finite}S(\G)<\infty.\end{align} 

\textbf{Part 2: 
weak lower semi-continuity.} 
 We pass from the positive trace class operator $\g$ to the Hilbert-Schmidt one, $\ka:=\sqrt \g$. Note that $\g\in I^{1,1}, \g\ge 0 \Longleftrightarrow \ka:=\sqrt \g \in I^{1,2}$.  Now, instead of the free energy functional \eqref{cF}, 
 we consider the equivalent functional 
  \begin{align} \label{F-def}
	F (\ka, \al, e) &:=\mathcal{F} (\G, a_b+e)\big|_{\g=\ka^2} - b^2 |\Omd|\\ 
 \label{F} & 
  =\Tr\big( \ka h_{a_b+e} \ka \big) + \frac12 \int \rho_\g(v* \rho_\g) \notag\\ &\qquad	 +\frac{1}{2} \Tr\big( \s^* v^\sharp \s \big) 
+ \int_{\Omega}  |\curl e|^2  
 - TS(\G)\big|_{\g=\ka^2}
\end{align}
on the space $I^{1,2}\times I^{1,2}\times \vec H^{1}$ with  the norm $\|(\ka, \al, e)\|_{(1)} := \|\ka \|_{I^{1,2}} + \|\al \|_{ I^{1,2}}+\|e \|_{H^{1}}$ and with the side conditions $0\le \G\big|_{\g=\ka^2}\le 1$ and $\tr\ka^2=\nu$.
We will keep the notation $\mathcal{D}^1_\nu$ for $I^{1,2}\times I^{1,2}$ with these side conditions.

By Proposition \ref{pro:LowerBound}, 
 we find that, for $\tr \g=\nu$ and $0\le \eta \le \one$, 
\begin{align} 
F (\ka, \al, e)	
\label{F-lower-bnd}& \geq 
\frac14 c' \min  ( [\|(\ka, \al)\|_1^2/ \nu]^r, \|(\ka, \al)\|_1^2)+ c\|e\|_{H^1}^2 - C,\end{align}
where $\|(\ka, \al)\|_1^2 := \|\ka\|_{I^{1,2}}^2+ \|\al\|_{I^{1,2}}^2$, and for suitable constants $c, C > 0$, with $C$ depending on $\nu$, $T, \|v\|_{\infty}$ and $|\lat|$.

\begin{lemma}\label{lem:Flsc}
The functional $F (\ka, \al, e)$ is weakly lower semi-continuous in  
 $I^{1,2}\times I^{1,2} \times \frachvec^1$. 
\end{lemma}
\begin{proof} We study the functional $F (\ka, \al, e)$ term by term. 
For the first term on the r.h.s. of \eqref{F}, 
with $a=a_b+e$ and $h_{a}:= -\Delta_{a}$, we write 
\begin{align} \label{hAgamSplit} 	
	& \Tr(h_{a} \gamma) =  \Tr((-\Delta_{a_b}) \gamma) + 2i\Tr(e \cdot \nabla_{a_b} \gamma)+ \Tr(|e|^2 \gamma). 
\end{align}
Since the first term on the r.h.s. of \eqref{hAgamSplit} satisfies $ \Tr((-\Delta_{a_b})  \gamma)=\| \ka \|_{I^{1, 2}}^2$ 
 and is quadratic in $\ka$, it is $\|\cdot \|_{I^{1,2}}$-weakly lower semi-continuous.

For the second term on the r.h.s. of \eqref{hAgamSplit}, 
 we let $e, e'\in \vec H^1$ and estimate the difference $\Tr(e \cdot \nabla_{a_b} \gamma) - \Tr(e' \cdot \nabla_{a_b} \gamma')$.
 We write 
\begin{align}\label{eq1}
	&  \Tr(e \cdot \nabla_{a_b} \gamma) - \Tr(e' \cdot \nabla_{a_b} \gamma')\notag \\
	&=  \Tr((e-e') \cdot \nabla_{a_b} \gamma) - \Tr(e' \cdot \nabla_{a_b} (\gamma-\gamma')). 
	\end{align}
For the first term on the r.h.s., 
 letting $c:=e-e'$, we claim that 
\begin{align}\label{c-n-gam-est'}|\Tr(c \cdot \nabla_{a_b} \gamma)|  \ls \|c\|_{H^s} \|\gamma\|_{I^{1,1}},\ s<1. \end{align}
 To prove this inequality, we recall that $M_b:= \sqrt{-\Delta_{a_b}}$  and write $\Tr(c \cdot \nabla_{a_b} \gamma)=\Tr(M_b^{-1}c \cdot \nabla_{a_b}M_b^{-1} M_b \gamma M_b)$ and use a standard trace class estimate to obtain $|\Tr(c \cdot \nabla_{a_b} \gamma)| \ls \|M_b^{-1}c \cdot \nabla_{a_b}M_b^{-1} \| \|M_b \gamma M_b\|_{I^{0,1}}$. Next, we use the boundedness of $\nabla_{a_b}M_b^{-1}$, the relative bound  
  $\|M_b^{-1}c  \|  \ls \|c\|_{H^s}, s<1$ (see \eqref{Sob-ineq1} of Appendix \ref{sec:den-est}), and the relation $\|M_b \gamma M_b\|_{I^{0,1}}= \|\gamma\|_{I^{1,1}}$ to find \eqref{c-n-gam-est'}.  (Recall that  $\| \kappa\|_{I^{s,2}}= \|\gamma\|_{I^{s,1}}^{1/2}$.)

For $\g$ and $\g'$  non-negative, we claim the following estimate for the second term on the r.h.s. of \eqref{eq1} with $c=e'$:
\begin{align}\label{c-n-gam-est}	|\Tr(c \cdot \nabla_{a_b} &(\gamma - \gamma'))|   \notag\\
	&\ls \|c\|_{H^t}( \|\kappa\|_{I^{1,2}} 
	+\|\kappa'\|_{I^{1,2}}) \|\kappa-\kappa'\|_{I^{s,2}} ,\ s, t<1,\end{align}
where $\kappa:=\g^{1/2}$ and  $\kappa':=(\g')^{1/2}$. To prove this estimate,  we write $\g=\kappa^2, \g'={\kappa'}^2$ to expand 
\begin{align}\label{eq2}	
(\gamma - \gamma') & 
=\kappa (\kappa-\kappa')+(\kappa-\kappa')\kappa'.\end{align}  
Now, we use  the boundedness of $\nabla_{a_b}M_b^{-1}$  and the  relative bound $\|M_b^{-s}c  \|  \ls \|c\|_{H^t}, s, t<1$ (see \eqref{Sob-ineq1} of Appendix \ref{sec:den-est}),  and the relations $\|M_b^s \kappa\|_{I^{0,2}}= \|\kappa M_b^s \|_{I^{0,2}}= \| \kappa\|_{I^{s,2}}$, to find
\begin{align}\label{ineq1}	|\Tr(c \cdot \nabla_{a_b} \kappa (\kappa-\kappa')) | &= |\Tr(M_b^{-s}c \cdot \nabla_{a_b}M_b^{-1} M_b\kappa (\kappa-\kappa')M_b^{s})| \notag\\
	&\ls \|c\|_{H^t}\|\kappa\|_{I^{1,2}} \|\kappa-\kappa'\|_{I^{s,2}} ,\ s, t<1.\end{align}
For the second term on the r.h.s. of \eqref{eq2}, we use
 the  relative bound  $\|M_b^{-1}c \cdot \nabla_{a_b}M_b^{-s} \|  \ls \|c\|_{H^1},$  for any $ s<1$  (see \eqref{Sob-ineq2} of Appendix \ref{sec:den-est}) to find
\begin{align}\label{ineq2}	|\Tr(c \cdot \nabla_{a_b} (\kappa-\kappa') \kappa') |& = |\Tr(M_b^{-1}c \cdot \nabla_{a_b}M_b^{-s} M_b^{s}(\kappa-\kappa') \kappa')M_b)| \notag\\
	&\ls \|c\|_{H^1}\|\kappa'\|_{I^{1,2}} \|\kappa-\kappa'\|_{I^{s,2}}. 
	\end{align}
 The last two estimates yield \eqref{c-n-gam-est}.


Applying \eqref{c-n-gam-est'} and \eqref{c-n-gam-est} to the terms on the r.h.s. of \eqref{eq1}, we find, for $3/4 < s<1$,
\begin{align}\label{2nd-term-cont}
	 |\Tr(c \cdot \nabla_{a_b} \gamma) & - \Tr(c' \cdot \nabla_{a_b} \gamma')|\ls  \|c-c'\|_{\vec{\mathfrak{h}}^{s}}  \|\gamma \|_{I^{1, 1}}\notag \\
	&\quad \quad+ \|c'\|_{\vec H^{s}}  
	( \|\kappa\|_{I^{1,2}} 
	+\|\kappa'\|_{I^{1,2}}) \|\kappa-\kappa'\|_{I^{s,2}} ,\ s, t<1,\end{align}
where $\kappa:=\g^{1/2}$ and  $\kappa:=(\g')^{1/2}$. 

 Now, we use a standard result for Sobolev spaces, $\vec H^{s'}$ is compactly embedded in $\vec H^{s}$, for any $s' > s$, and, perhaps, a less standard one, that  $I^{s', 2}$ is compactly embedded in $I^{s, 2}$, for any $s' > s$.  
 (One shows the latter fact by passing to the integral kernels and using a standard  Sobolev embedding result.) 

Now, let $\{(\ka_n, \al_n, e_n)\}$ be a weakly convergent sequence in $I^{1, 2} \times I^{1, 2} \times \frachvec^1$ 
 and denote its limit by $(\ka_*, \al_*, a_*)$. Then, by above, it converges strongly in  $\mathcal{D}^s_\nu \times\frachvec^{s},\ s<1$. Hence,  
 we have by \eqref{2nd-term-cont}, 
\begin{align}\label{2nd-term-cont'}
	 |\Tr(e_n \cdot \nabla_{a_b} \gamma_n) & - \Tr(e_* \cdot \nabla_{a_b} \gamma_*)|   \ra 0,\ n\ra \infty, 
	\end{align}
where, as usual, $\g_n=\ka_n^2$ and $\g_*=\ka_*^2$.

Finally, we consider the difference $\Tr(|e|^2 \gamma) - \Tr(|e' |^2 \gamma' )$ due to the last term in \eqref{hAgamSplit}. We decompose
\begin{align}\label{eq3}\Tr(|e|^2 \gamma) -& \Tr(|e' |^2 \gamma' )= \Tr(|e|^2 (\gamma - \gamma' )) + \Tr((|e|^2-|e' |^2)\gamma').\end{align}
For the first term on the r.h.s. we  claim the following estimate
\begin{align}\label{c-n-gam-est''}|\Tr(|e|^2  (\gamma - \gamma' )) | &\ls \|e\|_{H^t}^2 (\|\kappa\|_{I^{1,2}} +\|\kappa'\|_{I^{1,2}}) \|\kappa-\kappa'\|_{I^{s,2}} ,\ s <1,\end{align}
where  $\kappa:=\g^{1/2}$ and  $\kappa':=(\g')^{1/2}$. We use again \eqref{eq2} and
 the  relative bound  $\|M_b^{-s}|e|^2M_b^{-t} \|  \ls \|e\|_{H^r}, s + t>2(1-r),$ (see \eqref{Sob-ineq3} of Appendix \ref{sec:den-est}) to find, similarly to \eqref{ineq1} and \eqref{ineq2}, 
\begin{align}\label{ineq3}|\Tr(|e|^2  (\gamma - \gamma' )) | &=	|\Tr(|e|^2( \kappa (\kappa-\kappa') + (\kappa-\kappa') \kappa') |\\&\le |\Tr(M_b^{-s}|e|^2M_b^{-1} M_b\kappa (\kappa-\kappa')M_b^{s})| \notag\\
&\qquad +|\Tr(M_b^{-1}|e|^2M_b^{-s} M_b^{s}(\kappa-\kappa') \kappa' M_b)| \notag\\
	&\ls \|e\|_{H^t} (\|\kappa\|_{I^{1,2}} +\|\kappa'\|_{I^{1,2}}) \|\kappa-\kappa'\|_{I^{s,2}} ,\ s <1,\end{align}
 which gives  \eqref{c-n-gam-est''}. 
 
 Finally, similarly to \eqref{c-n-gam-est'}, we find for the second term on the r.h.s. of \eqref{eq3},
 \begin{align}\label{c-n-gam-est'''}|\Tr((|e|^2-|e' |^2)\gamma')|  \ls (\|e\|_{H^s} + \|e'\|_{H^s})\|e - e'\|_{H^s} \|\gamma\|_{I^{1,1}},\ s<1. \end{align} 
  
 Now, Eqs  \eqref{eq3},    \eqref{c-n-gam-est''} and \eqref{c-n-gam-est'''} imply 
\DETAILS{term $\Tr(|e|^2 \gamma)$ on the r.h.s. of \eqref{hAgamSplit}. 
We show \begin{align}\label{e2-gam-est}
	|\Tr(|e|^2  \gamma)| \le   \|e\|_{H^{t}}^2 \|\g\|_{ I^{s,1}}, s>2(1-1/q), t>1/q, q<2.\end{align}
 Firstly,	 we use \eqref{den-def} to obtain $\Tr(|e|^2  \gamma)=\int |e|^2\rho_\g$, which gives $|\Tr(|e|^2  \gamma)|\ls \|e\|_{L^{p}}^2 \|\rho_\g\|_{L^{q}}$. This and the estimate  $\|\rho_\g\|_{L^{q}}\ls \|\gamma\|_{I^{s,1}}, s>2(1-1/q),$  proven in Appendix \ref{sec:den-est} (see  \eqref{den-ineq4}) yield \eqref{e2-gam-est}.  Estimate \eqref{e2-gam-est} implies}
\begin{align}|\Tr(|e_n|^2 \gamma_n) - \Tr(|e_*|^2 \gamma_*)| 
		&\lesssim  
\|e_n\|_{H^1}^2 (\|\kappa_n\|_{I^{1,2}} +\|\kappa_*\|_{I^{1,2}}) \|\kappa_n-\kappa_*\|_{I^{s,2}}\notag\\
		 &\quad+ \|\gamma_*\|_{I^{1,1}} (\|e_n\|_{H^s} + \|e_*\|_{H^s})\|e_n - e_*\|_{H^{s}},\end{align}
for $s< 1$, and therefore, as above the r.h.s. converges to $0$. Hence,  the first term on the r.h.s. of \eqref{F} is  is weakly lower semi-continuous.

  For the second term on the r.h.s. of \eqref{F}, as in the proof of Proposition \ref{prop:F1-min}, we use the inequalities \eqref{DI-est} and \eqref{rho-est} to show that it is weakly lower semi-continuous in $I^{1, 2}$. The third term on the r.h.s. of \eqref{F} is  quadratic in $\alpha$ and therefore it is continuous in $I^{1, 2} \times I^{1, 2}  
 \times \frachvec^1$ since $v\in L^\infty$. It follows that it is weakly lower semi-continuous in $I^{1, 2} \times I^{1, 2}  
 \times \frachvec^1$.
\DETAILS{we note that for any $\alpha_1,\alpha_2$,
\begin{align}
	& \Tr(\alpha_1^* v^\sharp  \alpha_1)  - \Tr(\alpha^*_2 v^\sharp  \alpha_2) \\
	&= \Tr((\alpha_1-\alpha_2)^* v^\sharp  \alpha_1)  - \Tr(\alpha^*_2 v^\sharp  (\alpha_1 - \alpha_2))
\end{align}
Since $v$ is bounded, we see that the above terms are bounded by, up to constants,
\begin{align}
	\| \alpha_1 - \alpha_2 \|_2 \|v\|_\infty (\|\alpha_1\|_2+\|\alpha_2\|_2) \lesssim \| \alpha_1 - \alpha_2 \|_2 \|v\|_\infty
\end{align}
where the last inequality follows since $0 \leq \alpha^* \alpha \leq \gamma \leq 1$ in $\mathcal{D}$.} 

The third term on the r.h.s. of \eqref{F}, $\int_{\Omd} |\curl e|^2$, is clearly convex. So its norm lower semi-continuity is equivalent to weak semi-continuity. Since it is clearly $\frachvec^1-$norm continuous, it is $\frachvec^1-$weakly lower semi-continuous.

 Hence all the terms on the r.h.s. of the expression \eqref{F} for $
  F (\ka, \al, e)$, save  $-T S(\G)$, are lower semi-continuous under the convergence indicated. 
   The  lower semi-continuity of the latter term is proven in Lemma \ref{lem:Slsc} of Appendix \ref{sec:entropy}. Hence $ F (\ka, \al, e)$ is lower semi-continuous, which completes the proof of Lemma \ref{lem:Flsc}.\end{proof}

  Finally, we observe that the set $\cD^1_\nu \times \frachvec^1$ is  closed in $I^{1, 2} \times I^{1, 2}  
 \times \frachvec^1$ under the weak convergence.

We continue with the proof of Theorem \ref{thm:ExistMin}. With the results above, the proof of existence of a minimizer is standard. Let  $\{(\ka_n, \al_n, e_n)\}$ be a weakly convergent sequence in $
\mathcal{D}^1_\nu \times \frachvec^1$ 
 which is a minimizing sequence for $F(\ka, \al, e) 
 $. By  \eqref{F-lower-bnd}, 
the norm $\|(\ka_n, \al_n, e_n)\|_{(1)} (:= \|\ka_n \|_{I^{1,2}} + \|\al_n \|_{ I^{1,2}}+\|e_n \|_{\frachvec^{1}})$ (see the line after \eqref{F}) is bounded uniformly in $n$. By Sobolev-type embedding theorems,  $(\ka_n, \al_n, e_n)$ converges strongly in $\cD^s_\nu\times \frachvec^{s}$ for any $s<1$ and by 
  the Banach-Alaoglu theorem, 
   $(\ka_n, \al_n, e_n)$ converges weakly in $\cD^1_\nu\times \frachvec^{1}$. Denote the limit by $(\ka_*, \al_*,e_*)$. Since, by Lemma \ref{lem:Flsc}, $F$ is  lower semi-continuous, we have
\begin{align}
	\liminf_{n \rightarrow \infty} F(\ka_n, \al_n, e_n) \ge F(\ka_*, \al_*,e_*).
\end{align}
Hence, $(\ka_*, \al_*, e_*)$ is indeed a minimizer. This proves the existence of a minimizer of the functional \eqref{F-def} and therefore of 
$\mathcal{F}(\G,a)$.

Now, we establish properties of the minimizer $(\ka_*, \al_*, e_*)$. By \eqref{S-finite}, we have $\Tr g(\G_*)<\infty$, where $\G_*$ corresponds to $(\ka_*, \al_*)$.  
The statement, that $g(\G_*)$ is trace class, follows from \eqref{S-finite} and the fact that $g(\G)\ge 0$. 

Furthermore, if we restrict ourselves to even $(\ka, \al, e)$, i.e. to $(\ka, \al, e)$ satisfying \eqref{even}, then the minimizer $(\ka_*, \al_*, e_*)$ is also even.


Now, since a minimizing sequence $e_n$ converges to $e_*$ strongly in $\vec{\mathfrak{h}}^s$ for any $s<1$, we have, by the magnetic flux quantization \eqref{mf-quant} for $e_n$, the convergence of $e_n$ to $e_*$ and by the Stokes theorem, that $\frac{1}{2\pi} \int_{\Omd} \curl a_* = c_1(\rho)\in \Z$, where, recall, $a_*=a_b +e_*$. 

Finally, the last property of the minimizer $(\ka_*, \al_*, e_*)$ is shown in the following 
\begin{lemma}\label{lem:eta*EVs}
$0$ and $1$ are not eigenvalues of $\G_*$.  Consequently, $0 < \G_* < 1$.\end{lemma}
\begin{proof} 
\DETAILS{We use throughout the proof that for any $v \in \frachb \times \frachb$, normalized, and and operator $A$ on $\frachb$, we have $\Tr AP_v = \lan v, Av\ran$ {\bf(this looks strange as $A$ acts on $\frachb$ and $v\in \frachb \times \frachb$)}, where $P_v$ is the projection onto $v$, and we  write $\lan v, A v\ran$ instead of $\Tr A P_v$.} 
We assume for the sake of contradiction that $\G_*$ has the eigenvalue $0$. Hence $1$ is also an eigenvalue, since, if $\G_* x = 0$, then  \eqref{Gam-prop} implies $\G_* J\bar x = J\bar x$. 
Let $P_x$ and $P_{J\bar x}$ be the orthogonal projections onto the subspaces spanned by $x$ and $J\bar x$.   Define
\begin{align}
	\G' := P_x - J\bar P_x J^* = P_x - P_{J\bar x}. \label{eqn:G'-Px-JPJ}
\end{align}
Let $\G_\e:=(1+2\e/\nu)^{-1}(\G_*+\epsilon \G')$. Since $x$ and $J\bar x$ are orthogonal, it is straightforward to see that $\Tr\G_\e =\nu$, $\G_\e\in \cD^1_\nu$ and $0 \le\G_\e \le 1$.  

Below, we 
  write $F(\G)$ for $F(\ka, \al, e)$ and compute $F(\G_\e)$. By \eqref{F}, 
we have 
\begin{align}\label{F-exp1} F(\G_\epsilon) - F(\G_*)= - T(S(\G_\epsilon) - S(\G_*))+O(\e). 
\end{align} Using that $S(\G) = \Tr(s(\G))$, with $s(\G ) = -2\G\ln\G$, and writing $\Tr(s(\G))$ in an orthonormal eigen-basis which includes the eigenvectors $x$ and $J\bar x$, we find 
\[S(\G_\epsilon) - S(\G_*)= \lan x, s(\G_\epsilon)x \ran + \lan J\bar x, s(\G_\epsilon)J\bar x \ran.\] Then we use Jensen's inequality to find
 \begin{align}\label{S-ineq1}S(\G_\epsilon) - S(\G_*)\ge  s(\lan x, \G_\epsilon x \ran) + s(\lan J\bar x, \G_\epsilon J\bar x \ran).\end{align} 
Now, using the definition $\G_\e:=(1+2\e/\nu)^{-1}(\G_*+\epsilon \G')$ and the relations $\G_* x = 0$ and $\G_* J\bar x = J\bar x$, 
 we compute $\lan x, \G_\epsilon x \ran=(1+2\e/\nu)^{-1}\e$ and $\lan J\bar x, \G_\epsilon J\bar x \ran=(1+2\e/\nu)^{-1}(1+\e)$. This, together with \eqref{S-ineq1} and the definition $s(\lam):= - 2\lam  \ln \lam$, yields
\begin{align}\label{S-ineq2}S(\G_\epsilon) - S(\G_*)&\ge  s((1+2\e/\nu)^{-1}\e) + s((1+2\e/\nu)^{-1}(1+\e))\notag\\
&= -2 \epsilon  \ln \epsilon  + O(\epsilon).\end{align}
(The nonlinear $\epsilon \ln \epsilon$ term in \eqref{S-ineq2} is due to the fact that $x\ln(x)$ is not differentiable at $0$.) This, together with \eqref{F-exp1}, implies
\begin{align}\label{F-exp2} F(\G_\epsilon) &\le  F(\G_*) + 2T  \epsilon  \ln \epsilon  + O(\epsilon). 
\end{align}
Since $ \ln \epsilon <0$ and $|\epsilon \ln \epsilon| \gg \epsilon$, we conclude that $F(\G_\epsilon) <  F(\G_*)$
which contradicts minimality of $\G_*$. We conclude that $\G_*$ has a trivial kernel; hence, it has a trivial $1$-eigenspace. Consequently, $0 < \G_* < 1$. 
\end{proof}
This completes the proof of Theorem \ref{thm:ExistMin}.\end{proof} 

\subsection*{Acknowledgments}
The first author is very grateful for Almut Burchard's kind help and suggestions. The second author is grateful to Volker Bach, S\'ebastien  Breteaux, Thomas Chen and J\"urg Fr\"ohlich for enjoyable collaboration and both authors thank Rupert Frank and  Christian Hainzl, for stimulating discussions.

The authors are indebted to the anonymous referees for many useful suggestions and remarks.

 The authors' research is supported in part by NSERC Grant No. NA7901. During the work on the paper, the authors enjoyed the support of the NCCR SwissMAP.

 \appendix


\section{Entropy} \label{sec:entropy}

In this appendix we prove the differentiability and expansion of the entropy functional, which we recall here
\begin{align}
\label{S-def''}	& S(\G) := \Tr(s(\G))= \Tr(g(\G)),\\
\label{g-s-def''}	&  g(\G): = -\G\ln\G - (1-\G)\ln(1-\G),\ s(\G ):= -2\G\ln\G.
\end{align}
 This is used in the next two appendices in order to prove  Theorem \ref{thm:BdG=EL} and Propositions \ref{prop:FT''} and  \ref{prop:FT-expan-order2}.
Let 
$d S(\G  ) \G' :=\partial_\epsilon S(\G  + \epsilon \G') \mid_{\epsilon=0} $. We have

\begin{proposition}\label{prop:S-deriv} Let $\G \in \mathcal{D}^1_\nu$ be such that 
 $g(\G) := -\G \ln \G - (1-\G)\ln(1-\G)$ is trace class and $\G'$ satisfy \eqref{Gam'-cond}. Then $S$ is $C^1$ and its derivative is given by
\begin{align}\label{dS}	& d S(\G  ) \G' = \Tr( g'(\G ) \G')= \Tr( s'(\G ) \G').
\end{align}
\end{proposition}
\begin{proof}
By  \eqref{S-def''}, it suffices for us to prove the proposition for $s(\G) = -\G\ln(\G)$.
Denote $\G'':=\G +\epsilon \G'$. We write 
\begin{align}\label{S-AB} S(\G'') - S(\G ) &= -\Tr (\G(\ln\G''-\ln \G) - \e \G'(\ln \G''-\ln \G) - \e \G'\ln \G)\\
\label{S-AB}&=: A+B - \e \Tr (\G'\ln \G).
	\end{align} 
Using the formula $\ln a -\ln b= \int_0^\infty [(b + t)^{-1} - (a + t)^{-1}] d t$ and the second resolvent equation, we compute
\begin{align}
	A &:= -\Tr (\G(\G''-\ln \G)) \notag\\
	&= \int_0^\infty \Tr \{ \G[(\G'' + t)^{-1} - (\G+t)^{-1}]\} dt\notag\\
			&  = -\int_0^\infty \Tr \{ \G (\G+ t)^{-1} \e\et' (\G''+t)^{-1}\} dt\notag\\
	&  = -\int_0^\infty \Tr \{ \G (\G+ t)^{-1} \e\et' (\G+t)^{-1}\} dt \notag\\
	& \quad -\int_0^\infty \Tr \{ \G (\G+ t)^{-1} \e\et'  (\G+ t)^{-1} \e\et' (\G''+t)^{-1}\} dt.
		 \label{A-comp2}
	\end{align}
Similarly, we have
\begin{align}
	B &:= -\Tr (\e \G'(\G''-\ln \G)) \notag\\
	&= \int_0^\infty \Tr \{ \e \G'[(\G''+ t)^{-1} - (\G+t)^{-1}]\} dt\notag\\
			&  =  -\int_0^\infty \Tr \{ \e \G' (\G+ t)^{-1} \e\et' (\G''+t)^{-1}\} dt.
			\label{B-comp2}
	\end{align}
Combining the last two relations with \eqref{S-AB}, we find
\begin{align}\label{S-exp2} & S(\G +\epsilon \G')- S(\G )=  \e S_1+  \e^2 R_2, 
\\
\label{S1-orig}&S_1:= -\Tr \G'\ln \G - \int_0^\infty \Tr \{ \G (\G+ t)^{-1} \et' (\G+t)^{-1}\} dt, \\
&R_2:= \int_0^\infty \Tr \{ \G (\G+ t)^{-1} \et'  (\G+ t)^{-1} \et' (\G''+t)^{-1} \notag\\
\label{R2-orig}& \qquad - \G' (\G+ t)^{-1} \et' (\G''+t)^{-1}\} dt.	
	\end{align} 
The estimates below show that the integrals on the r.h.s. converge. Computing the integral $ \int_0^\infty \Tr \{ \G (\G+ t)^{-1} \et' (\G+t)^{-1}\} dt= \int_0^\infty \Tr \{ \G (\G+ t)^{-2} \et' \} dt=\Tr  \et' $ in the expression for $S_1$ and transforming  the expression for $R_2$, 
we obtain
\begin{align}\label{S1}  S_1:=& -\Tr \{ \G'\ln \G + \et' \} , \\ 
\label{R2} R_2:= & -\int_0^\infty \Tr \{ t (\G+ t)^{-1} \et'  (\G+ t)^{-1} \et' (\G''+t)^{-1} \} dt. 
	\end{align}

The proofs of convergence of \eqref{S1-orig} and \eqref{R2} are similar. We consider the case of \eqref{R2}.
We estimate the integrand on the r.h.s. of \eqref{R2}. 
 We have 
\begin{align}& | \Tr \{  \G' (\G+ t)^{-1} \et' (\G''+t)^{-1}\} |\notag\\ 
&\qquad \qquad  \qquad  \le  \| \et' (\G+t)^{-1}\|_{I^2} \| \et' (\G''+t)^{-1}\|_{I^2} 
\end{align} 
Now, we show that the factors on the r.h.s. are $L^2(dt)$. By the second condition in \eqref{Gam'-cond} on $\et'$, we have 
\begin{align}\| \et' (\et^\#+t)^{-1}\|_{I^2} &\le \| \et (\one -\et^\#) (\et^\#+t)^{-1}\|_{I^2} \notag\\ & 
 \le \| \xi^\# (\xi^\# +t)^{-1}\|_{I^2},
 \end{align}
where $\et^\#$ is either $\et$ or $\et''$ and $\xi^\#:=\et^\# (\one -\et^\#)$. Let $\mu_n$ be the eigenvalues of the operator $\xi^\#:=\et^\# (\one -\et^\#)$. Then we have 
\begin{align}  \label{eqn:1stVar}	&  \|\xi^\# (\xi^\# +t)^{-1}\|_{I^2}^2 =\sum_n  \mu_n^2 (\mu_n+t)^{-2} , 
\end{align}
and therefore 
\begin{align}\int_0^\infty  \| \xi^\# (\xi^\#+t)^{-1}\|_{I^2}^2 dt & =\int_0^\infty   \sum_n  \mu_n^2 (\mu_n+t)^{-2}  dt\notag\\
&=\sum_n  \mu_n = \Tr   \xi^\#.\end{align}  
Since $ \et(\one -\et)$ and $ \et''(\one -\et'')$ are trace class operators, this proves the claim and, with it, the convergence of the integral on the l.h.s.. Similarly, one shows the convergence of the other integrals.

To sum up, we proved the expansion \eqref{S-exp2} with $S_1$ given by \eqref{S1}, which is the same as \eqref{dS},  and $R_2$ bounded as $|R_2|\ls 1$.
In particular, this implies that  $S$ is $C^1$ and its derivative is given by \eqref{dS}.
\end{proof}

\begin{proposition}\label{prop:S-expan-order2} 
$S(\G) := \Tr(g(\G)) $ is $C^3$ at $\G_{Tb}$ w.r.t. perturbations $\et'$ satisfying \eqref{Gam'-cond}.
 Moreover,  we have
\begin{align}\label{S-expan} 
	S(\G_{Tb}+ \e\G') =& S(\G_{Tb}) + \e S'(\G_{Tb})\G' + \frac12 \e^2 S''(\G',\G')  + O(\epsilon^3), 
 \end{align}
where 
$S'(\G_{Tb})\G':=\Tr(g'(\G_{Tb})\G')$, $S''(\G',\G')$ is a quadratic form given by
\begin{align}\label{S''}  	
	S''(\G',\G')=& \frac12  \int_0^\infty \Tr \{  (\G+ t)^{-1} \et'  (\G+ t)^{-1} \et' \} dt
\end{align} 
 and the error term is uniform in $\G'$ and is bounded by $\e^3 \Tr(\G_{Tb}(1-\G_{Tb}))$.  For 
$\G'=\phi(\al)$, the quadratic term becomes 
\begin{align}\label{S''-al} &S''(\G',\G') = -\Tr\left(\bar{\alpha} K_{T b} \alpha \right),\\  & K_{T b}:=\frac{1}{T} \frac{h_{T b}^L+h_{T b}^R}{\tanh(h^L_{T b}/T) + \tanh(h^R_{T b}/T)}, \end{align}
 where $h_{T b}:=h_{\gamma_{Tb} a_b \mu}\equiv h_{\gamma_{Tb} a_b } - \mu$, with 
  $h_{\g a}$ defined in \eqref{Lam}.
\end{proposition}
\begin{proof} For the duration of the proof we omit the subindex $Tb$ in $\G_{Tb}$. Recall 
\eqref{S-AB}-\eqref{B-comp2} and continuing computing $A$ and $B$ in  \eqref{A-comp2}-\eqref{B-comp2} in the same fashion as in the derivation of these equations, we find
\begin{align}
	A 	=&  \int_0^\infty \Tr \{ \G (\G+ t)^{-1} \e\et' (\G+t)^{-1}\} dt \notag\\
	&- \int_0^\infty \Tr \{ \G (\G+ t)^{-1} \e\et'  (\G+ t)^{-1} \e\et' (\G+t)^{-1}\} dt\notag\\
	&+ \int_0^\infty \Tr \{ \G (\G+ t)^{-1} \e\et'  (\G+ t)^{-1} \e\et' (\G+ t)^{-1} \e\et' (\G''+t)^{-1}\} dt, \label{A-comp3}
	\end{align}
and
\begin{align}
	B 			=&  \int_0^\infty \Tr \{ \e \G' (\G+ t)^{-1} \e\et' (\G+t)^{-1}\} dt \notag\\
	&- \int_0^\infty \Tr \{ \e \G' (\G+ t)^{-1} \e\et'  (\G+ t)^{-1} \e\et' (\G''+t)^{-1}\} dt . \label{B-comp3}
	\end{align}
Combining the last two relations with \eqref{S-AB} and recalling the computation of $S_1$, we find
\begin{align}\label{S-exp3}  S(\G +\epsilon \G')&- S(\G )=  \e S_1+  \e^2 S_2+  \e^3 R_3\\
	S_2:=&- \int_0^\infty \Tr \{ \G (\G+ t)^{-1} \et'  (\G+ t)^{-1} \et' (\G+t)^{-1} \notag\\
	&- \G' (\G+ t)^{-1} \et' (\G+t)^{-1}\} dt,\notag\\
		R_3:= &\int_0^\infty \Tr \{ \G (\G+ t)^{-1} \et'  (\G+ t)^{-1} \et' (\G+ t)^{-1}  \et' (\G''+t)^{-1}\notag\\
	&- \G' (\G+ t)^{-1} \et'  (\G+ t)^{-1} \et' (\G''+t)^{-1}\} dt. 
	\end{align} 
Transforming  the expressions for $S_2$ and $R_3$, we obtain
\begin{align}\label{S2'}  
	S_2=& \int_0^\infty \Tr \{ t (\G+ t)^{-1} \et'  (\G+ t)^{-1} \et' (\G+t)^{-1}\} dt, \\
	\label{R3}	R_3	= &-\int_0^\infty \Tr \{ t (\G+ t)^{-1} \et'  (\G+ t)^{-1} \et' (\G+ t)^{-1}  \et' (\G''+t)^{-1}\} dt.
	\end{align} 
Estimates similar to those done after \eqref{R2} show that the integrals on the r.h.s. converge.
This proves the expansion \eqref{S-exp3} with $S_1$ and $S_2$ given by \eqref{S1}, which is the same as \eqref{dS}, and \eqref{S2'} and  $R_3$ bounded as $|R_3|\ls 1$.
Identifying  the quadratic form $S_2$ with  $S''(\G',\G')$, we arrive at the expansion \eqref{S-expan}.

Before computing $S_2\equiv S''(\G',\G')$, we find a simpler representation for it. Integrating the r.h.s. of  \eqref{S2'} by parts, we find
\begin{align}\label{S''-comp}  	S''=& \int_0^\infty \Tr \{ t (\G+ t)^{-2} \et'  (\G+ t)^{-1} \et' \} dt \notag \\
		=& \int_0^\infty \Tr \{  (\G+ t)^{-1} \et'  (\G+ t)^{-1} \et' \} dt \notag\\
	&- \int_0^\infty \Tr \{ t (\G+ t)^{-1} \et'  (\G+ t)^{-2} \et' \} dt.
	\end{align}
But by the cyclicity of the trace the last integral is equal to the first one and therefore we have \eqref{S''}. Eq \eqref{S-sym} gives 
 \begin{align} 	\label{S''-sym}  S''(\G',\G')=& \frac14  \int_0^\infty \Tr \{  [(\G+ t)^{-1} \et'  (\G+ t)^{-1} \notag\\ &+  (\one - \G+ t)^{-1} \et'  (\one - \G+ t)^{-1}]\et' \} dt.
	\end{align}
	
Now, we use \eqref{S''-sym} to compute to $S''$ for $\G' = \phi(\alpha)$. First,  we recall that $\et=\G_{Tb}$ and observe that for $\G' = \phi(\alpha)$,
\begin{align}
\notag		  \Tr((\G_{Tb}+t)^{-1}\G'(\G_{Tb}+t)^{-1}\G') &= 2\Tr((\g_{Tb}+t)^{-1}\al(\one- \bar\g_{Tb}+t)^{-1}\al) \\
		  &= 2\Tr(((x+t)^{-1}(y+t)^{-1} \alpha) \bar{\alpha} )
\end{align}
where 
$x$ and $y$ are regarded as operators acting on $\alpha$ from the left by multiplying by $\gamma_{Tb}$ and 
from the right, by $1-\bar{\gamma}_{Tb}$. Putting this together with 
 a similar expression for the second term on the r.h.s. of \eqref{S''-sym} and performing the integral in $t$, we obtain 
\begin{align} \label{S''1}	& S''(\G',\G')= -\Tr\left[ \bar{\alpha} K  (\alpha) \right],\\  
& \label{K1} K  :=\frac{\log(x)-\log(y)}{x-y} + \frac{\log(1-x) - \log(1-y)}{(1-x)-(1-y)}, 
\end{align}
 with $x$ acting on the left and $y$ acting on the right. \eqref{K1} can be written as 
 \begin{align}\label{K-expr} 
 K&=- \frac{\log(x^{-1}-1)-\log(y^{-1}-1)}{x-y}. 	
\end{align}

Recalling that $\gamma_{Tb} = f_T(h_{T b})=(1+e^{2 h_{T b}/T})^{-1}$ 
 and therefore $x^{-1}-1=e^{2 h^L_{T b}/T}$ and $y^{-1}-1=e^{- 2 \bar h^R_{T b}/T}$, we see that
\begin{align}
	K=& \frac{1}{T}\frac{h_{T b}^L+\bar h_{T b}^R}{(1+e^{h^L_{T b}/T})^{-1} + (1+e^{\bar h^R_{T b}/T})^{-1}}, \end{align}
which, together 
 the hyperbolic functions identities, $(1+e^{h})^{-1}=\frac12 (1-\tanh h)$ and $(1+e^{-h})^{-1}=\frac12 (1+\tanh h)$,  gives \eqref{S''-al}. \end{proof} 
By the definition of the G\^ateaux derive and the Hessian and Proposition \ref{prop:S-expan-order2}, we have

\begin{corollary}\label{cor:S'-S''} 
 We have
\begin{align}\label{S'}
	d S(\G_{Tb})\G':=\Tr(g'(\G_{Tb})\G'), 
 	\end{align} 
 and, for $\G'=\phi(\al)$ and with $h_{T b}:=h_{\gamma_{Tb}, a_b}$, 
\begin{align}\label{S''-K} &S''(\G_{Tb})\phi(\al) = -\phi(K_{T b} \alpha),\ 
 K_{T b}:=\frac{1}{T} \frac{h_{T b}^L+h_{T b}^R}{\tanh(h^L_{T b}/T) + \tanh(h^R_{T b}/T)}. \end{align}
\end{corollary}

Our next result on the entropy is the following

\begin{lemma}\label{lem:Slsc}
The functional $- S(\G)$ is weakly lower semi-continuous in $\mathcal{D}^1_\nu$. 
\end{lemma}
\begin{proof}
We use an idea from \cite{L2} which allows to reduce the problem to a finite-dimensional one. We use \eqref{S-relatS}, to pass from $- S(\G)$ to 
 the relative entropy,  $S(\G | \G_0)$, defined in \eqref{RelatEntropy}, with $\G_0$ of the form \eqref{Gam0Phi}, with $\Tr \gamma_0< \infty$ and s.t. $S(\G_0)<\infty$. 
  By \eqref{S-relatS}, $S(\G | \G_0)\ge 0$. 
Moreover,
\begin{align}\label{S-deco}
	 S(\G) = S(\G_0) - S(\G|\G_0) - \Tr[(\G-\G_0)\ln \G_0].
\end{align}
We choose $\G_0$ so that $(\G-\G_0)\ln \G_0$ is trace class and the term $\Tr[(\G-\G_0)\ln \G_0]$ is weakly lower semi-continuous. We take 
\begin{align} \label{eqn:H0choice}
	& \G_0 = f_{T}(M), \, M := \text{diag}(\sqrt{-\Delta_{a_b}}, -\overline{\sqrt{-\Delta_{a_b}}}). 
\end{align}
Since $ f_{T}(h)=\rbrac{e^{ h/T}+1}^{-1}$, 
we see that that
\begin{align} \label{eqn:ln-eta0}
	0 \leq -\ln \eta_0 = \ln (1+e^{ M/T}) \lesssim & 1+  M/ T.
\end{align}
This estimate and (\ref{eqn:H0choice}) show that $\Tr[(\G-\G_0)\ln \G_0]$ is $\hat I^1$-norm continuous ($\hat I^1-$norm is defined in \eqref{eta-norm}).  
Indeed, writing $ \ln\eta_0=(1+  M/ T)(1+  M/ T)^{-1} \ln\eta_0$ and using \eqref{eqn:ln-eta0}, we find
\begin{align} \label{eta-ln-eta0-est}	\|(\G-\G_0) \ln\eta_0\|_{I^1} \le \|(\G-\G_0) (1+M/ T)\|_{I^1}&\|(1+  M/ T)^{-1} \ln\eta_0\|\notag\\
& \ls \|\eta-\G_0\|_{I^{1,1}}.
\end{align}
This completes the proof of the claim that $\Tr[(\G-\G_0)\ln \G_0]$ is $\hat I^1$-norm continuous.

Furthermore, since this term is affine in $\G$, it is convex. Thus it is weakly lower semicontinuous.


Now, following \cite{L2}, let $s_\lambda(A|B) = \lambda^{-1}(s(\lambda A + (1-\lambda) B) - \lambda s(A)- (1-\lambda) s(B))$ and write
\begin{align}
	S(\G_n | \G_0) +& \Tr(\G_0 - \G_n) 
	= \sup_{\lambda \in (0,1)} \Tr( s_\lambda(\G_n | \G_*)).
\end{align}
Since the entropy function $s$ is concave, $s_\lambda(A|B) \geq 0$ for any $A,B$. For any non-negative operator $T$ on $L^2(\Om)$, $\Tr_{L^2(\Om)} T = \sup_P \Tr_{L^2(\Om)} PT$ where the sup is taken over all finite rank projections. Hence, we may write
\begin{align}
	S(\G_n | \G_0) +& \Tr(\G_0 - \G_n) 
	= \sup_{\lambda \in (0,1)} \sup_{P} \Tr( Ps_\lambda(\G_n | \G_0)
\end{align}
where the $\sup_P$ is taken over all finite rank projections $P$. It follows that for any $\lambda \in (0,1)$ and any finite rank projection $P$,
\begin{align}
	S(\G_n | \G_0) + \Tr(\G_0 - \G_n) \geq \Tr( Ps_\lambda(\G_n | \G_0)).
\end{align}
Since $\G_n \rightarrow \G_*$ in $\| \cdot \|_{(0)}$ (hence in operator norm) and $-x\ln x$ is continuous on $[0,1]$, we see that
\begin{align}
	 s_\lambda(\G_n | \G_0) \rightarrow   s_\lambda(\G_* | \G_0) \end{align}
in the operator norm. 
In particular, for any finite dimensional projection $P$, 
\begin{align}
	\Tr(P  s_\lambda (\G_n | \G_0)) \rightarrow \Tr( P   s_\lambda(\G_* | \G_0)).
\end{align}
Consequently,
\begin{align}
	\liminf_{n \rightarrow \infty} S(\G_n | \G_0) +& \Tr(\G_0 - \G_n) \geq \Tr( Ps_\lambda(\G_* | \G_0)).
\end{align}
Now taking $\sup_{\lambda \in (0,1)}$ and $\sup_P$ and using that $\Tr(\G_0 - \G_n)=0$, by condition \eqref{Gam'-cond}, we see that
\begin{align}
	\liminf_{n \rightarrow \infty} S(\G_n | \G_0) \geq \Tr( s(\G_* | \G_0)),
\end{align}
which implies the desired statement. 
  \end{proof}

As an aside not used in this paper, we compute the Hessian, $\p_{\g\g}S(\g_{T b})$, of $S$ w.r.t. diagonal perturbations,
\begin{align}\label{d-gam'} d(\g'):=\left( \begin{array}{cc} \g' & 0 \\ 0  & -\bar\g' \end{array} \right). 
 \end{align}
 $\p_{\g\g}S(\g_{T b})$ is defined by 
\begin{align}[\lan\g', \p_{\g\g}S(\g_{T b})\g'\ran:=\p^2_\e S(\G_{Tb}+ \e d(\g'))\big|_{\e=0}.\end{align} 
We have
\begin{proposition}\label{prop:S''-gam} 
 The hessian operator $\p_{\g\g}S(\g_{T b})$ is given by 
\begin{align}\label{S''-gam} \p_{\g\g}S(\g_{T b})= 
\frac1T\frac{h_{T b}^L-h_{T b}^R}{\tanh(h^L_{T b}/T) - \tanh(h^R_{T b}/T)}. \end{align}
\end{proposition}
 \begin{proof}Our starting point in the formula \eqref{S''} and hence we begin with the computation of the term $\int_0^\infty \Tr[(\G+ t)^{-1} \et'  (\G+ t)^{-1} \et'] dt$, with $\G' = d(\g')$, where $d(\g')$ 
 denotes the perturbation in $\gam$ given by
\begin{align}\label{d-gam'} d(\g'):=\left( \begin{array}{cc} \g' & 0 \\ 0  & - \bar\g' \end{array} \right). 
 \end{align}

 First,  we recall that $\et=\G_{Tb}$ and observe that for $\G' = d(\g')$,
\begin{align}
		  \Tr( &(\G_{Tb}+t)^{-1}\G'(\G_{Tb}+t)^{-1}\G') =  \Tr((\g_{Tb}+t)^{-1}\g'(\g_{Tb}+t)^{-1}\g'\\
		  &+ (\one- \bar\g_{Tb}+t)^{-1}\bar \g'(\one- \bar\g_{Tb}+t)^{-1}\bar \g') \\
		  &= \Tr(([(x+t)^{-1}(x'+t)^{-1}  +  (\one-x+t)^{-1} (\one -x'+t)^{-1} ]\g')\g'),
\end{align}
where the last follows from $\Tr(A) = \Tr(\bar{A})$ for self-adjoint operators, and 
$x$ and $x'$ are regarded as operators acting on $\g'$ from the left by multiplying by $\gamma_{Tb}$ and 
from the right, by $\gamma_{Tb}$. 
Performing the integral in $t$, we obtain $S''(\G',\G')= -\Tr\left[ \bar{\g'} K'  (\g') \right],$ where the operator $K'$ 
is given by
 \begin{align}
	  K':=\frac{\log(x)-\log(x')}{x-x'}+ \frac{\log(\one -x ) - \log(\one -x')}{(\one -x )-(\one -x' )}, 
\end{align}
 with $x$ acting on the left and $x'$ acting on the right. Clearly, $K'$ is identified with $\p_{\g\g}S(\g_{T b})$. Rewrite the operator  $K'$ as 
 \begin{align}\label{S''-expr}	K'	&= -\frac{\log(x^{-1}-1)-\log({x'}^{-1}-1)}{x-x'}. 
\end{align}

Recalling that $\gamma_{Tb} = g^\sharp(h_{T b}/T)=(1+e^{2 h_{T b}/T})^{-1}$, where $h_{T b}:=h_{\gamma_{Tb}, a_b}$, we see that
\begin{align}
	K =& \frac{1}{T}\frac{h_{T b}^L-h_{T b}^R}{(1+e^{h^L_{T b}/T})^{-1} - (1+e^{h^R_{T b}/T})^{-1}}, \end{align}
which, together with \eqref{S''-expr} and the hyperbolic functions identities, $(1+e^{h})^{-1}=\frac12 (1-\tanh h)$ and $(1+e^{-h})^{-1}=\frac12 (1+\tanh h)$,  gives \eqref{S''-gam}. \end{proof}

\section{Energy 
functional: Proof of Theorem \ref{thm:BdG=EL}} \label{sec:energy}

The proof of Theorem \ref{thm:BdG=EL} consists of three parts: 1) differentiability of $F_T$, 
 2) identification of the BdG equations with the Euler-Lagrange equation of $F_T$, and 3) showing minimizers of $F_T$ among the set $\mathcal{D}^1_\nu \times \frachbvec^1$ are critical points.

\textbf{Part 1: differentiability.} We consider first the variation $\G  + \epsilon \G'$ for $\epsilon > 0$ small and perturbations 
 satisfying \eqref{Gam'-cond}, 
  Note that such $\G'$ satisfies, for $\epsilon$ small enough,
 \begin{align}
	0 \leq \G  + \epsilon \G' \leq 1.
\end{align}
Let $d_\G F_T(\G , a  ) \G' :=\partial_\epsilon F_T(\G  + \epsilon \G' , a ) \mid_{\epsilon=0} $, if the r.h.s.  exists. From \eqref{energy}, it is easy to see that $E(\G, a)$ is Fr\'echet differentiable and 
\begin{align}\label{dcF} 
	d_\G E(\G , a  ) \G' =  \Tr(\Lambda(\G , a ) \G'). 
\end{align}
Hence it suffices to prove the  Fr\'echet differentiability of $S(\G  )$. This is done in Appendix \ref{sec:entropy} above.

 Differentiability of $F_T$ with respect to $a$ is standard and can be easily done. The only two terms in $F_T$ that depend on $a$ are $\Tr((-\Delta_a)\g)$ and $\frac{1}{2}\int |\curl a|^2$. The first term can be differentiated by using $-\Delta_{a_0+a'} = (-\Delta_{a_0})^2 -2 a' (-i\nabla-a_0) + |a'|^2$ while the second term is differentiable by standard variational calculus. Hence the differentiability of $F_T$ follows from \eqref{FT-def}, \eqref{dcF} and Proposition \ref{prop:S-deriv}.																		
\medskip

\textbf{Part 2: Euler-Lagrange equation.} 

Now, we show that if $0 < \G < 1$ and  $d_\G F_T(\G,a)\G' = 0$ and $d_a F_T(\G,a)a' = 0$ for all $\G'$ on $\frach \times \frach$ satisfying 
\eqref{Gam'-cond} and $a' \in \frachvec^1$, then $(\G,a)$ satisfies the BdG equations \eqref{BdG-eq-t}-\eqref{Amp-Maxw-eq}. 

We start with $d_{\G} F_T\G'=0$ for all $\G'$ satisfying \eqref{Gam'-cond}. 
First, we construct explicitly a dense subset of perturbations $\G'$ satisfying \eqref{Gam'-cond}.
For a critical point $(\G,a),  0< \G < 1$, we define a reference unit vector 
$v_0 = (1,0)^T \in \frach \times \frach$. 
We 
note that the difference in norm of $v_0$'s two components is simply
\begin{align}
	0 \not= 1 = \|(v_0)_1\|_{\frach}^2 - \|(v_0)_2\|_{\frach}^2 = \lan v_0, S v_0 \ran = \Tr(SP_{v_0}).
\end{align}

For simplicity and without loss of generality, we assume that $v_0$ is in the image of $\G(1-\G)$ since $0 < \G < 1$ (i.e. its range is dense). We define $V \subset \frach \times \frach$ as
\begin{align}
	V = \{ v :& \, \|v\|_2=1, \, v = \G(1-\G) \xi,\ \xi \in \frach \times \frach \}. \label{eqn:def-v}
\end{align}
For each $v\in V$, we define
\begin{align}
	\G'_v = (P_v - P_{J \bar v}) - \frac{\Tr SP_v}{\Tr SP_{v_0}} (P_{v_0} - P_{J \bar v_0})\, ,
\end{align}
where $P_x$ is the orthogonal projection onto $x$ and $J$ is the complex structure in \eqref{Gam-prop}. 
\begin{lemma} $\G'_v$ satisfies \eqref{Gam'-cond}. \end{lemma} 
\begin{proof} To prove the first condition of \eqref{Gam'-cond}, we only prove it for $P_v - P_{J \bar v} = P_v - J P_{\bar v} J^*$ since this condition is real linear. (Note that $S$ is self-adjoint, so $\lan v, S v\ran =\Tr SP_v$ is real for all $v$.) We note that $J^*=J^{-1} = -J$ and has only real number components. Let $\mathcal{C}$ denote the complex conjugation. It follows then that
\begin{align}
	J^*(P_v - JP_{\bar v}J^*)J =& J^*P_v J - P_{\bar v} \\
		=& -\mathcal{C}(P_v - J^* \bar P_{v} J)\mathcal{C} \\
		=& -\mathcal{C}(P_v - JP_{\bar v}J^*)\mathcal{C} \, .
\end{align}
This proves the first condition in \eqref{Gam'-cond}.

To prove the second condition in \eqref{Gam'-cond}, it suffices to show that $P_v$ satisfies this condition for every $v \in \frach \times \frach$ since $J$ is unitary. For any $v = \G (1-\G)\xi, x \in \frach \times \frach$, we note that
\begin{align}
	\|P_v x\| = |\lan v, x\ran| = |\lan \G(1-\G)\xi, x \ran | \ls \|\xi\|_2\|\G(1-\G)x\|_2 .
\end{align}
This shows that $(\G')^2 \leq C [\G (1-\G )]^2$. This proves the second condition in \eqref{Gam'-cond}.

Finally, we prove the last condition in \eqref{Gam'-cond}. For any unit norm $v \in \frach \times \frach$,
\begin{align}
	\Tr S_1(P_v - JP_{\bar v}J^*) =& \Tr S_1P_v - \Tr S_1JP_{\bar v} J^* \\
		=& \Tr SP_v - \Tr J^*SJP_{\bar v} .
\end{align}
We note that $J^*S_1J = S_2 := \text{diag}(0,1)$. Hence,
\begin{align}
	\Tr S_1(P_v - JP_{\bar v}J^*) =& \Tr SP_v. 
\end{align}
It follows that
\begin{align}
	\Tr S_1 \G'_v = \Tr SP_v - \frac{\Tr SP_v}{\Tr SP_{v_0}} \Tr SP_{v_0} = 0.
\end{align}
This proves that  $\G'_v$ satisfies the last condition in \eqref{Gam'-cond}. \end{proof}

We show that  $0 < \G < 1$ and  $d_\G F_T(\G,a)\G' = 0$ for all $\G'$  satisfying \eqref{Gam'-cond} imply $\Lambda(\G ,a ) - \mu S -Tg'(\G ) = 0$ for some $\mu$ and $S = \text{diag}(1,-1)$ (see Proposition \ref{prop:stat-sol}). 
First note that \eqref{FT-def}, \eqref{dcF}, \eqref{dS} and \eqref{S-sym} yield that
\begin{align}\label{detaFT}
	d_\G F_T(\G,a) \G'= \Tr \left[A\G' \right],
\end{align}
where $A := \Lambda(\G ,a ) -Tg'(\G )$. If $(\G ,a )$ is  a critical point, then, for all $v \in V$, it satisfies
\begin{align}
	\tr(A\G'_v) = 0 \, .
\end{align}
 Since $A$ is in the tangent space of all the $\G$ such that $J^*\G J = 1-\bar \G$. We note that $A$ also satisfies  the first condition in \eqref{Gam'-cond}. It follows that
\begin{align}
	0 =& \Tr(A \G'_v) = \Tr(A P_v) - \Tr(AJ\bar P_v J^*) \\
			&- \frac{\Tr SP_v}{\Tr SP_{v_0}}  (\Tr(A P_v) - \Tr(AJ\bar P_v J^*)) \\
		=& \Tr(A P_v) - \Tr(J^*AJ\bar P_v ) \\
			&- \frac{\Tr SP_v}{\Tr SP_{v_0}}  (\Tr(A P_{v_0}) - \Tr(J^*AJ\bar P_{v_0} )) \\
		=& \Tr(A P_v) + \Tr(\bar A \bar P_v ) \\
			&- \frac{\Tr SP_v}{\Tr SP_{v_0}}  (\Tr(A P_{v_0}) + \Tr(\bar A \bar P_{v_0} )) \\
		=& 2\Tr(A P_v) - \frac{\Tr SP_v}{\Tr SP_{v_0}} (2\Tr AP_{v_0}).
\end{align}
We conclude that
\begin{align}
	\Tr AP_v = \frac{\Tr AP_{v_0}}{\Tr SP_{v_0}} \Tr SP_v =: \mu \Tr SP_v
\end{align}
for all $v \in V$. We note that $\mu$ is real since $A$ and $S$ are self-adjoint. Since $0 < \G < 1$, the linear space spanned $V$ is dense. We conclude that $A$ is a multiple of $S$, which we denote by $\mu$. This shows that
\begin{align}
	0 = A - \mu S. 
\end{align}
We conclude that $(\G,a)$ solves the first BdG equation, \eqref{Gam-eq}.

Now, we consider the equation $d_a F_T(\G,a)a' = 0$. As was mentioned above, one can easily show that 
\begin{align}
	d_aF_T(\eta, a)a' = 2\int dx v \cdot a' ,
\end{align}
where $v:=\curl^* \curl a - j(\gamma, a)$ and the perturbation $a' \in \frachvec^1$ is divergence free and mean zero. Hence, to conclude that $v=0$, 
we have to show that $v:=\curl^* \curl a - j(\gamma, a)$ is divergence free and mean zero. 
 (Indeed, any vector field $e$ can be written as $e=a'+\n g +c$, where $a'$ is divergence free and mean zero and $c$ is a constant, and therefore \[\int_\Om v e=\int_\Om v a'- \int_\Om \divv v g+ c\int_\Om v .\] 
 So, if  $v$ is divergence free and mean zero and $\int_\Om v a'=0$, for every $a'$ divergence free and mean zero, then $v=0$.)

 Clearly, the term $\curl^* \curl a$ is divergence free and mean zero.  So we show that $j(\gamma, a)$ is divergence free and mean zero. For the first property, we use the fact that our free energy functional is invariant under gauge transformation. In fact, it suffices to use the gauge invariance of the first line in \eqref{energy}, $E_1(T^{\rm gauge}_{t\chi}(\G, a))=E_1(\G, a)$, where  
 $E_1(\G, a) :=\Tr\big((-\Delta_a) \gamma \big) +\Tr\big((v* \rho_\g) \gamma \big) -\frac12\Tr\big((v^\sharp \g) \gamma \big)$. It gives
\begin{align}
	0 = \partial_t \mid_{t=0} E_1(T^{\rm gauge}_{t\chi}(\eta, a)),
\end{align}
for all $\chi \in H^1_{\rm loc}$ 
which are $\mathcal{L}$-periodic. Using the cyclicity of trace, 
we compute this explicitly 
\begin{align}\label{J-comp'}	0 
		=& \Tr(\Re(2i\nabla_a \gamma) \cdot \nabla \chi) + \Tr([\gamma, h_{\gamma,a}] \chi).
\end{align}
 Since $(\G,a)$ solves the BdG equation, we have that $[\Lambda(\G,a), \G] = 0$. Taking the upper left component of this operator-valued matrix equation, we see that
\begin{align}
	[h_{\gamma,a}, \gamma] + (v^\sharp \alpha) \bar{\alpha} - \alpha (v^\sharp\bar{\alpha}) = 0.
\end{align}
Since $v(x)=v(-x)$, we conclude that the integral kernel of $(v^\sharp \alpha) \bar{\alpha} - \alpha (v^\sharp\bar{\alpha})$,
\begin{align}
	\int (v(x-z)-v(z-y))\alpha(x,z)\bar{\alpha}(z,y) dz,
\end{align} 
 is zero on the diagonal. Thus, the same conclusion holds for $[\gamma, h_{\gamma,a}]$. Consequently, $\Tr([\gamma, h_{\gamma,a}] \chi)=0$ and we conclude, by \eqref{J-comp'}, that 
\begin{align}\label{J-comp}
	0 = -\Tr(\Re(2i\nabla_a \g) \cdot \nabla \chi) =
	 \int_{\Omd} j(\gamma,a) \cdot \nabla \chi=0.
\end{align}
Since this is true for every $\chi \in H^1_{\rm loc}$ which are $\mathcal{L}$-periodic, it follows that $\div j(\gamma, a) = 0$. 

To show that $j(\gamma, a)$ is mean zero, 
we use, that by our assumptions, $\g$ is even and $a$ is odd. Since for any operator $A$, $u^{\rm refl}\den[A] = \den[u^{\rm refl}Au^{\rm refl}]$ where $(u^{\rm refl}f)(x) = f(-x)$, this shows that $j(\gamma, a)$ is odd. Hence so is $v:=\curl^* \curl a - j(\gamma, a)$ and therefore $v:=\curl^* \curl a - j(\gamma, a)=0$.

Since $\div a = 0$, we may replace $\curl^* \curl$ by $-\Delta$. Hence, the elliptic regularity theory shows that $a \in \frachbvec^2$. This completes the proof. $\Box$

\medskip

\textbf{Part 3: minimizers are critical points.} 
For a minimizer $(\G,a)$, we have that $d_\G F_T(\G,a)\G'$, $d_a F_T(\G,a)a' \geq 0$. 
 Since $\vec H^1$ is linear, $a' \in \frachvec^1$ if and only if $-a' \in \vec H^1$. So $d_a F_T(\G,a)a' = 0$ for all $a \in \frachvec^1$. Similarly, we note that $\G'$ satisfies the assumption \eqref{Gam'-cond} 
 if and only if $-\G'$ satisfies the same requirement. Hence we conclude that 
	$0 = d F_T(\G  , a ) \G' ,$ 
which completes the proof. $\Box$


\section{Proof of the existence of solution to \eqref{xi-fp}} \label{sec:xi-fp} 

\begin{lemma} \label{lem:xi-fp}
Assume $\int v \geq 0$. 
 Then, for each $T > 0$ and $b = \frac{2\pi n}{|\Om|}$,  the fixed point problem \eqref{xi-fp}
   has a unique solution. 
\end{lemma}
\begin{proof} 
Let  $\lam:=\int v$.  We define the real function $g_T:\R \rightarrow \R$ by
\begin{align}
	g_T(\xi) :=  \lam \den(f_{T}(-\Delta_{ a_b } - \mu + \xi))(0),
\end{align}
where, recall,  
$\Delta_{\ab }$ acts on the space $\frach$ and is self-adjoint. Then  \eqref{xi-fp} can be rewritten as 
\begin{align}\label{xi-fp'}
	\xi = g_T(\xi). 
	\end{align}
First, we derive a more convenient formula for the function  $g_T$. Note that, since, recall $f_{T}(s)=(1+e^{s/T})^{-1}$ and, for $b = \frac{2\pi n}{|\Om|}$, the operator $-\Delta_{a_b}$ on $\frach$ has the  eigenvalues $b(2m+1), m=0, 1, \dots$, each of the same multiplicity $n$,  the trace $\Tr f_{T}(-\Delta_{a_b} - \mu + \xi)$ is finite and is smooth in $\xi$. 
Since, by Proposition \ref{prop:mt-invar-oprs}, $\den[f_{T}(-\Delta_{a_b} - \mu + \xi)]$ is constant and since $\int_{\Om}\den A=\Tr A$, we have that
\begin{align} \label{gT-expr}
	g_T(\xi) = \lam\frac{1}{|\Om|}\Tr f_{T}(-\Delta_{a_b} - \mu + \xi)
\end{align}
and that $g_T$ is a smooth function.

Since $f_{T}(s)=(1+e^{s/T})^{-1}$ is positive and $ f_{T}'(s)=-  T^{-1} e^{s/T}(1+e^{s/T})^{-2}$ is negative, \eqref{gT-expr} shows that  $\pm g_T(\xi)>0$ 
 and $\pm g_T'(\xi)<0$, if $\pm \lam>0$. Hence, for $\lam>0$, \eqref{xi-fp'} has a unique solution for every $T>0$ and $b = \frac{2\pi n}{|\Om|}$. 
\end{proof}

\section{Relative bounds and estimates on density} \label{sec:den-est}
 In this appendix we prove bounds on functions relative to the operator $M_b$ and estimates on density $\rho_\gamma$. Our first result is the following
\begin{lemma} \label{lem:Sob-est} We have the following Sobolev-type  inequalities
 \begin{align} \label{Sob-ineq1} &\|M_b^{-s}c M_b^{-t} \|  \ls \|c\|_{H^r}, s+t>1-r, \\
 \label{Sob-ineq2} &\|\|M_b^{-1}c \cdot \nabla_{a_b}M_b^{-t} \|  \ls \|c\|_{H^r}, t>1-r,\\ 
  \label{Sob-ineq3}&\|M_b^{-s}|c|^2M_b^{-t} \|  \ls \|c\|_{H^r}^2, s+t>2(1-r).
 \end{align}
where in the second estimate we assumed $\divv c=0$. 
 \end{lemma}
\begin{proof} We use the diamagnetic inequality $|M_b^{-s} f|\le M_{b=0}^{-s}|f|$ (see \cite{AHS}) to reduce the problem to the $b=0$ case. To estimate the r.h.s. we write $M_0^{-s}$ as the convolution, $M_0^{-s}u=G_s* u$, where $G_s(x)$ is the Fourier transform of $(1+|k|^2)^{-s/2}$, and use that $G_s(x)$ decays exponentially at infinity and has the singularity $\asymp |x|^{-2+s}$ at the origin. Hence $G_s\in L^t(\R^2), t<2/(2-s)$ and we can estimate by the Young inequality $\|G_s* u\|_{L^k}\ls \|G_s\|_{L^t}\|u\|_{L^q}, 1+1/k=1/t+1/q, t<2/(2-s),$ to obtain
\begin{align} \label{Sob-ineq5} &\|M_b^{-s}  f\|_{L^k}  \ls  \|f\|_{L^r}, 1/k+s/2>1/r, s<2.
 \end{align}

Now, to prove  \eqref{Sob-ineq1} and \eqref{Sob-ineq3}, we apply \eqref{Sob-ineq5} twice and the Young inequality to obtain
\begin{align} \label{Sob-ineq6} &\|M_b^{-s}c M_b^{-t} f\|_{L^2}  \ls  \|c M_b^{-t} f\|_{L^q} \ls  \|c\|_{L^p} \|M_b^{-t} f\|_{L^k} \ls  \|c\|_{L^p} \| f\|_{L^2}, 
 \end{align}
with $1/2+s/2>1/q=1/p +  1/k, 1/k+t/2=1/2$, which implies $s+t>2/p$. To obtain  \eqref{Sob-ineq1} and \eqref{Sob-ineq3}, we use the Sobolev inequalities $\|c\|_{L^p}  \ls  \|c\|_{H^r}, r>1-2/p$ and $\||e|^2\|_{L^p} =\|e\|_{L^{2p}}^2 \ls  \|e\|_{H^r}^2, r>1-1/p$, respectively.

To prove \eqref{Sob-ineq2}, we use, in addition,  $M_b^{-1}c \cdot \nabla_{a_b}M_b^{-s}=M_b^{-1}\nabla_{a_b} \cdot c M_b^{-s}+M_b^{-1} (\n c) M_b^{-s}$ and $\divv c=0$ to reduce the problem to \eqref{Sob-ineq1} with $s=0$. \end{proof}

\begin{lemma} \label{lem:den-est} Let $\g$ be a trace-class and positive operator and let $\ka:=\sqrt\g$. Then 
 \begin{align} \label{den-ineq1} &\|\rho_\gamma\|_{W^{s, 1}} \ls  \|\ka\|_{I^{s,2}} \|\ka\|_{I^{0,2}},  \\ 
 \label{den-ineq2} &\|\rho_\gamma\|_{W^{1, 1}} \ls ( \|\gamma\|_{I^{1, 1}}  \tr \gamma)^{1/2},\\ 
 \label{den-ineq3}&\|\rho_\gamma\|_{L^q} \ls  \|\gamma\|_{I^{1,1}}^{1-r}  (\tr \gamma)^{r},\ \forall r\in (0, 1),\\ 
 \label{den-ineq4}&\| \rho_{\g}\|_{L^{q}}\ls  \|\gamma\|_{I^{s,1}},\ s>2(1-1/q). 
 \end{align}
 \end{lemma}
\begin{proof}
 We use \eqref{den-def} and $\p\rho_{\g}=\rho_{[\p_{a_b}, \g]}$ to obtain  \begin{align}\notag\|\p\rho_\gamma\|_{L^{ 1}}=\sup_{\|f\|_\infty=1}|\int f \p\rho_\gamma|=\sup_{\|f\|_\infty=1}|\tr ( f [\p_{a_b},\gamma])|\ls \|[\p_{a_b},\gamma]\|_{I^{0,1}}.\end{align} 
 Now, writing $\gamma=\ka^2$ and combining $\p_{a_b}$ with one of the $\ka$'s, we estimate furthermore $\|[\p_{a_b},\gamma]\|_{I^{0,1}} \ls\|\ka\|_{I^{1,2}} \|\ka\|_{I^{0,2}}$. Then we interpolate between $s=0$ and $s=1$ to get the first inequality.

 For the second inequality, we 
 let $\gamma=\kappa^2$ and write $\rho_\gamma(x)=\int \kappa(x, y) \kappa(y, x)$. It is not hard to see that
\[\|\rho_\gamma\|_{W^{1, 1}} \ls \|\kappa\|_{I^{1, 2}}  \|\kappa\|_{I^{0, 2}}=( \|\gamma\|_{I^{1, 1}}  \tr \gamma)^{1/2}. \]

The second inequality, together with $\int_\Om \rho^q\le (\int_\Om \rho^{(q-v)/(1-v)})^{1-v}  (\int_\Om \rho)^{v}$,  $v< 1< q$, and a Sobolev inequality $\|\rho_\gamma\|_{L^p} \ls \|\rho_\gamma\|_{W^{s, 1}}^{1-v/q}, s> 2(1-1/p)$,
  implies \eqref{den-ineq3}.

The first inequality, together with the Sobolev  inequality, $\|\rho_\g\|_{L^{3/2}}\ls \|\rho_\gamma\|_{W^{s, 1}}, s>2(1-1/q) $,  gives \eqref{den-ineq4}. \end{proof}

\begin{lemma} \label{lem:e2gam-est} We have for any $r\in (0,1)$,
\begin{align}\label{e2gam-est'''}
	0\le \Tr(|e|^2 \gamma)\ls  \|\g\|_{I^{1, 1}}^{1-r}\  (\tr \gamma)^{r}  \|e\|_{H^1}^2.  
\end{align}
 \end{lemma}
\begin{proof} We write $\Tr(|e|^2 \gamma)=\int_\Om |e|^2 \rho_\gamma $ and apply to this the H\"older and Sobolev inequalities to obtain
\[0\le \Tr(|e|^2 \gamma)\ls (\int_\Om |e|^{2p})^{1/p} (\int_\Om \rho_\gamma^q)^{1/q}\ls \|e\|_{H^1}^2\|\rho_\gamma\|_{L^{q}}, 
\]
where $1/p+ 1/q=1.$ This inequality together with  inequality  \eqref{den-ineq3} above gives  \eqref{e2gam-est'''} with $r\in (0,1)$. 
 \end{proof}
 We present another way to prove \eqref{e2gam-est'''} assuming the non-abelian interpolation inequality 
\begin{align}\label{e2gam-est''}
	 \|\ka \|_{I^{s, 2}} \ls  \|\ka\|_{I^{1, 2}}^s \|\ka\|_{I^{0, 2}}^{1/2-s}.\end{align}
 We use $\g=\ka \ka$ and write, for any $s, t>0$,
\begin{align}\label{e2gam-est''}
	0\le \Tr(|e|^2 \gamma)&= \Tr( M_b^{-s}|e|^2  M_b^{-t} M_b^{t}\ka \ka  M_b^{t})\\
	&\ls \| M_b^{-s}|e|^2 M_b^{-t}\|  \|M_b^{t}\ka \|_{I^{2}}  \|\ka M_b^{s}\|_{I^{2}}. \end{align}
The last  inequality, together with \eqref{e2gam-est''} and relative bound \eqref{Sob-ineq3}, gives  \eqref{e2gam-est'''} with $r\in (0,1)$. 	
	

\section{Quasifree reduction} \label{sec:qf-reduct} 

In general, a many-body evolution can be defined on states  (i.e. positive linear (`expectation') functionals) on the CAR or Weyl CCR algebra ${\frak W}$ over, say, Schwartz space $\cS(\R^d, \C^2)$. Elements of this algebra are operators acting  on  the 
  fermionic/bosonic Fock space $\cF$.\footnote{For a more detailed description, see \cite{BBCFS} and for the background \cite{BR}.} 

 To fix ideas we concentrate on spin $1/2$ fermions. Details for bosons could be found in \cite{BBCFS}. Let $\psi(\x)$ and $\psi^*(\x) $, where $\x:=(x, \s),\ x\in \R^d, \s\in \{\frac12, - \frac12\}$, the spatial and spin variables, be the annihilation and creation operators satisfying the canonical anticommutation relations. Given a quantum Hamiltonian $H$ on $\cF$, the evolution of states is  given by the 
von~Neumann-Landau equation \begin{align}\label{vNeum-eq} 
    i \partial_t\om_t(A) = \om_t([A, H]) \,,\ \forall A\in   \frak W,  
\end{align}
 where $\om_t$ is the state at time $t$.   (We leave out technical questions such as a definition of $\om_t([A, H])$ as $[A, H]$ is not in $\frak W$.) 

Let $N:=\int d\x \; \psi^*(\x) \psi(\x)$, where $\int d\x:=\sum_\s\int d x$, be the particle number operator. We distinguish between (a) {\it confined} systems with $\om(N)<\infty$ 
 and (b) {\it thermodynamic} systems with  $\om(N)=\infty$. In the former case the states are given by density operators on $\cF$, i.e. $\om (A)=\tr (A D)$, 
 where $D$ is a positive, trace-class operator on $\cF$ with unit trace (see e.g.  \cite{BLS}, Lemma 2.4). 

As the evolution \eqref{vNeum-eq} is practically intractable, 
one is interested in manageable approximations. The natural and most commonly used ones are one-body ones, which trade the number of degrees of freedom for a nonlinearity. 

The most general one-body approximation is given in terms of quasifree states. 
 A quasifree state $\qf$ determines and is determined by the 
  expectations to the second order (for fermions, we may assume that $\qf (\psi (\x))=0$):
\begin{equation} \label{omq-mom} \begin{cases}
    \gamma (\x, \y) :=   \qf(\psi^{*}(\y) \, \psi(x)).\\ 
    \al (\x, \y):=   \qf (\psi(\x) \, \psi(\y)). 
\end{cases} \end{equation}
Namely, 
with the short-hand notation $\psi_j  :=\psi^{\sharp}(\x)$,  where $\psi^{\sharp}(\x)$ is either $\psi(\x)$ or $\psi^{*}(\x)$, the $n$-point expectations, $\qf(\psi_1 \cdots \psi_n)$,  are given by the Wick theorem as 
$\qf(\psi_1 \cdots \psi_n) =0$ for $n$ odd and, for $n$ even, as
\begin{equation}
\label{eqn:quasi-wick} 
\qf(\psi_1 \cdots \psi_n) =
\sum_{P_n}\ve(P_n) \prod_{J \in P_n}
 \qf(\psi_{i_1}, \psi_{i_{2}})\,, 
\end{equation}
where the $P_n$ are partitions of the ordered set $\{1, ..., n\}$ into
ordered subsets, $J$, of two elements and $\ve(P_n)$ is $+1$ or $-1$ depending on whether the permutation  $\{1, ..., n\}\ra (J_1, \dots, J_{n/2})$ is even or odd.
%


 Assuming $\qf$ is $SU(2)$ invariant, the spatial and spin variables separate as $ \gamma (\x, \y)= \gamma (x, y)$ and $\al (\x, \y)=\al (x, y)\chi (\s, \tau)$, where $\al (x, y)$ is symmetric under the interchange of $x$ and $y$ and $\chi (\s, \tau)$ is antisymmetric under the interchange of $\s$ and $\tau$ (see \cite{HaiSei}).

Let $\gamma$ and $ \al$ denote the operators with the integral kernels $\gamma (x,y)$ and $\s (x,y)$. One can now verify readily that they satisfy \eqref{gam-al-prop}. 
\DETAILS{that
\begin{equation} \label{gam-al-prop}\gamma=\gamma^*\ge 0\ \text{ and }\ \s^*=\bar\s, \end{equation} 
where $\bar\sigma =C \sigma C$ with $C$ being the complex conjugation. (For  confined systems, $\g$ and $\s$ are trace class and Hilbert-Schmidt operators, respectively, with $\tr(\g)<\infty$ giving  the particle number, while for thermodynamic systems, they are not, unless we place the system in a box.) The  operator  $\g$ can be considered as a one-particle density matrix of the system, so that $\tr \g =\int \g(x, x) dx$ is the total number of particles.  The  operator  $\s$  gives the coherence content of the state.} 

However, the property of being quasifree is not preserved by the dynamics  \eqref{vNeum-eq} and the main question here is how to project the true quantum evolution onto the class of quasifree  states. 

 One elegant way was proposed by Dirac and Frenkel (see \cite{Lub} for a book exposition and references and \cite{BenSokSol}, for a recent treatment). Another one is due to \cite{BBCFS}. Following \cite{BBCFS}, we define self-consistent approximation as the restriction of the many-body dynamics to quasifree states.
More precisely, 
we map the solution $\om_t$ of \eqref{vNeum-eq}, with an initial state $\om_0$, 
 to the family $\qf_t$ 
 of quasifree states satisfying 
\begin{align}\label{eq-vNeum-quasifree}
    i  \partial_t\qf_t( A) = \qf_t([ A, H]) \,   
\end{align}
for all 
observables $A$, which are at most quadratic in the creation and annihilation operators,  with an initial state $\qf_0$, which is the quasifree projection of $\om_0$. 
We call this map the \textit{quasifree reduction} of equation \eqref{vNeum-eq}. 

 Of course, we cannot expect  $\qf_t$ to be a good approximation of $\om_t$, if $\om_0$ is far from the manifold of quasifree states.

Evaluating \eqref{eq-vNeum-quasifree}  on monomials $A\in 
 \{\psi(\x),\psi^*(\x)\psi(\y),\psi(\x)\psi(\y)\}$ yields a system of coupled nonlinear
PDE's for  $(\phi_t,\gamma_t,\al_t)$. For the standard  many-body Hamiltonian, 
\begin{align}  \label{H} 
    H = \int d\x \; \psi^*(\x)h\psi(\x) 
    + \frac{1}{2}\int d\x d\y \; v(x, y)\psi^*(\x) \psi^*(\y)
    \psi(x) \psi(y) \,,
\end{align}
with $h:= -\Delta +V(x)$ acting on the variable $x$ and $v$ a pair potential of the particle interaction, defined on Fock  space, $\cF$,
these give the (time-dependent) \textit{Bogoliubov-\-de Gennes (BdG)Gennes} equations.  (In the case of bosons, we arrive at the (time-dependent) \textit{Hartree-\-Fock-\-Bogo\-liubov (HFB) equations}.) 

This is a straightforward, but non-rigorous, derivation of the important effective equations. To prove error bounds is another matter. There was 
a concerted effort in the last years with important progress and extensive literature. For a recent book and references, see \cite{BenPorSchl, BoccCenSchl, BossPavlPickSoff, GrillMached1, GrillMached2, GrillMached3, LewNamSchl, NamNap, PortRadSaffSchl, PetPick}.

Finally, we note that according to the BCS theory,  Hamiltonian \eqref{H} describes Cooper pairs of electrons with non-local, attracting interaction, $v(x, y)$, due to exchange of phonons.


\end{document}